\documentclass[11pt]{article}
\pdfoutput=1
\usepackage[margin=1in]{geometry}

\usepackage{amsmath,amssymb,amsfonts,amsthm,latexsym}
\usepackage{graphicx,xcolor}
\usepackage{enumitem}
\theoremstyle{plain}
\newtheorem{allcnt}{AllCnt}
\newtheorem{theorem}[allcnt]{Theorem}
\newtheorem{corollary}[allcnt]{Corollary}

\graphicspath{{pic/}}

\newcommand{\Prob}{\mathbb{P}}
\newcommand{\E}{\mathbb{E}}
\newcommand{\Var}{\text{Var}}
\newcommand{\Cov}{\text{Cov}}
\newcommand{\Ind}{\mathbb{I}}
\newcommand{\cE}{\mathcal{E}}
\newcommand{\cS}{\mathcal{S}}
\newcommand{\be}{{\bf e}}
\newcommand{\bw}{{\bf w}}
\newcommand{\bv}{{\bf v}}
\newcommand{\bz}{{\bf z}}

\newcommand{\ATE}{\text{ATE}}
\newcommand{\TTE}{\text{TTE}}
\newcommand{\AIE}{\text{AIE}}
\newcommand{\est}{\hat{\text{\normalfont est}}}

\newcommand{\indep}{\rotatebox[origin=c]{90}{$\models$}}

\begin{document}
	\title{Estimating Total Treatment Effect in Randomized Experiments with Unknown Network Structure}
	\author{Christina Lee Yu \footnote{School of Operations Research and Information Engineering, Cornell University, Ithaca, NY}
\,\, Edoardo M Airoldi \footnote{Department of Statistics, Operations, and Data Science, Fox School of Business, Temple University, Philadelphia, PA} 
\,\, Christian Borgs  \footnote{Department of Electrical Engineering and Computer Science, UC Berkeley, Berkeley, CA}
\,\, Jennifer T Chayes \footnote{Departments of Electrical Engineering and Computer Science, Statistics, Mathematics, and the School of Information, UC Berkeley, Berkeley, CA}
}
	\date{}
	\maketitle

	\begin{abstract}
Randomized experiments are widely used to estimate the causal effects of a proposed treatment in many areas of science, from medicine and healthcare to the physical and biological sciences, from the social sciences to engineering, to public policy and to the technology industry at large.
Here, we consider situations where classical methods for estimating the total treatment effect on a target population are considerably biased due to confounding network effects, i.e., the fact that the treatment of an individual may impact their neighbors' outcomes, an issue referred to as network interference or as non-individualized treatment response.
A key challenge in these situations, is that the network is often unknown, and difficult, or costly, to measure. 
In this paper, we characterize the limitations in estimating the total treatment effect without knowledge of the network that drives interference, assuming a potential outcomes model with heterogeneous additive network effects. This model encompasses a broad class of network interference sources, including spillover, peer effects, and contagion. 
Within this framework, we show that, surprisingly, given access to average historical baseline measurements prior to the experiment, we can develop a simple estimator and efficient randomized design that outputs an unbiased estimate with low variance. Our solution  does not require knowledge of the underlying network structure, and it comes with statistical guarantees for a broad class of models. 
We believe our results are poised to impact current randomized experimentation strategies due to its ease of interpretation and implementation, alongside its provable theoretical insights under heterogeneous network effects.
	\end{abstract}
	
\section{Introduction}
The measurement of treatment effects via randomized experiments is a fundamental tool used across all fields of scientific disciplines and beyond. For example, consider a public health campaign to increase public awareness of the importance of wearing masks during a global pandemic. The administrator in charge of running the public health campaign would like to determine which proposed banner ad would be most effective for displaying on a public billboard. In particular, the administrator would like to estimate the ``total treatment effect'', i.e. the change in behavior of the population at large that results from viewing the proposed banner ad. The ``total treatment effect'' is a causal effect, as it describes the change in behavior that is caused by the treatment. The experimental units in this example are the individuals in the population, and outcomes refer to some measurable behavior of individuals, such as whether an individual is wearing a mask or not at the grocery store.

The classical approach to estimating causal effects involves running a randomized experiment, where one randomly partitions the population into a treatment group and a control group. The treatment of interest is administered to each individual in the treatment group, and a placebo is administered to each individual in the control group. The causal effect is then approximated by the difference in measured outcomes or behaviors between the treatment and control groups after the treatment has been administered. This approach results in an efficient unbiased estimate for the desired causal effect under a critical assumption that the outcome of an individual is not affected by the treatment assignment of any other individual; this assumption is referred to in the literature as the Stable Unit Treatment Value Assumption (SUTVA) \cite{cox1958planning,rubin1978bayesian,Manski13}. 

Unfortunately, SUTVA is violated in many applications, as individuals are connected in a complex social network that mediates communication, influence, or spread of disease, resulting in network interference that couple the outcomes of individuals. The treatment of individual A may impact the outcome of individual B, violating SUTVA and introducing significant bias to the estimates resulting from the classical experimental approach of randomizing treatment and control uniformly over individuals, and estimating the difference in average outcomes of the treatment and control group. Furthermore, the network which mediates the interference effects is often unknown and difficult or costly to estimate. As a result, we need new theory for experimental design which can account for these network interference effects, and yet are simple and practical to implement. We illustrate a few motivating scenarios. 

Consider estimating the total treatment effect of a public health campaign to increase the use of masks in public during a pandemic. Suppose that individual A, a senior citizen, is shown the proposed banner ad, and thus decides to wear a mask in public. Individual A's behavior could cause a positive network effect on their friends,  
even though they did not see the original banner ad. In contrast, individual A's behavior may have a negative effect on a teenager, who may think that wearing masks must be for the elderly and thus not ``cool''. 

Consider a social media platform such as LinkedIn, which would like to estimate the total treatment effect of a proposed change in the recommendation engine for a user's news feed, with a goal of increasing user engagement. The change in engagement level of individual A as a result of being exposed to the proposed change could subsequently impact the engagement level of others in individual A's social network, resulting in a positive or negative network effect. Similar network effects arise in other types of communication networks beyond social media platforms, including mobile networks, email exchange networks, and collaboration communities.

Consider running a clinical trial to estimate the total treatment effect of a proposed vaccine for COVID-19, i.e. how much would the overall rate of cases contracted in the public at large decrease as a result of everyone receiving the vaccine. Since COVID-19 is transmitted via an underlying social contact network, the impact of individual A receiving the vaccine may not only reduce individual A's chance of contracting the disease, but also may reduce the risk of exposure of other individuals connected to A in the network. This network effect is heterogeneous as the frequency of time individual A spends with others in its contact network may vary.

\section{Problem Setup and Potential Outcomes Model}

Consider a finite population of $n$ individuals. We denote the treatment vector by $\bz = (z_1, z_2, \dots z_n) \in \{0,1\}^{[n]}$, where $z_i = 1$ if individual $i$ is assigned the treatment and $z_i = 0$ if individual $i$ is in the control group. Let ${\bf e}_i$ denote the standard basis vector which takes value one at coordinate $i$ and is zero everywhere else. 
As we consider the randomized experiment setting, we assume that the treatment vector $\bz$ is sampled from a prescribed distribution as determined by the experimental design; this distribution is referred to as the randomized design. In many practical applications, there is a limit on the fraction of individuals that can be assigned to the treatment group, whether because of a high cost for testing a new treatment, or due to safety considerations of limiting possible unknown adverse effects. Therefore, a desired solution involves proposing an estimator alongside a randomized design for which we can achieve consistent estimation while keeping the expected number of treated individuals low.

$Y_i(\bz)$ denotes the potential outcome of individual $i$ in the event that treatment vector $\bz$ is implemented. Only the outcomes for the implemented treatment vector $\bz$ are observed, and thus all other ``potential outcomes'' that would result from other realizations of the treatment vector are unobserved. 
Under the Stable Unit Treatment Value Assumption, the potential outcome of individual $i$ only depends on $z_i$ and not on the treatment of any other individual \cite{cox1958planning,rubin1978bayesian,Manski13}. Under this assumption, $Y_i(\bz) = Y_i(z_i {\bf e}_i)$ for all $\bz$. In the presence of general arbitrary network interference, the outcome of individual $i$ may depend on the full treatment vector.
Our results rely on the neighborhood interference assumption alongside joint assumptions of additivity of main effects and interference effects as defined in \cite{SussmanAiroldi17}, which we also refer to as the heterogeneous additive network effects assumption.

\subsection{Heterogeneous Additive Network Effects}
The neighborhood interference assumption posits that an individual's outcome can only depend on the treatment assignments of its direct neighbors in a specified network $\cE$ \cite{UganderKarrerBackstromKleinberg13, SussmanAiroldi17}. We will furthermore assume that the network is unknown. The joint assumptions of additivity of main effects and interference effects as defined in \cite{SussmanAiroldi17} impose that the potential outcomes satisfy
\[Y_i(\bz) = Y_i(0) + (Y_i(\be_i) - Y_i(0)) + \textstyle\sum_{k\in [n]} (Y_i(\be_k) - Y_i(0)).\]
This enforces that the outcome of each individual is affected by an additive term for each treated individual in the population, but this additive network interference effects can be fully heterogenous for each pair of individuals. 
By letting $\alpha_i$ denote the individual baseline outcome $Y_i(0)$, $\beta_i$ denote the individual direct effect $(Y_i(\be_i) - Y_i(0))$ and $\gamma_{ki}$ denote the additive network interference effect over the directed pair $(k,i)$ given by $(Y_i(\be_k) - Y_i(0))$, it follows that the potential outcomes model is equivalently represented as
\[Y_i(\bz) = \alpha_i + \beta_i z_i + \textstyle\sum_{k \in [n]} \gamma_{ki} z_k.\]
This model trivially satisfies the neighborhood interference assumption with respect to  the network edge set $\cE$ defined as the set of pairs $(k,i)$ where $\gamma_{ki} \neq 0$. The total number of model parameters is $2n + |\cE|$. Without imposing constraints on the sparsity of $\cE$, $|\cE|$ could be as large as $n (n-1)$.

The model at a glance looks similar to a linear model, yet a key distinction is that our model allows for fully heterogeneous network effects individualized to each edge, such that the number of parameters grows with the population size. This model is also referred to as the linear model in \cite{EcklesKarrerUgander17} or the saturated structural linear model in \cite{hu2022average}. 
Our model is significantly more expressive than typical linear models used in the empirical literature which impose that there are a fixed number of types of individuals which then share the same network effect coefficients, such that the number of model parameters is fixed with respect to the population size.

\subsection{Target Estimand}
There are many potential estimands of interest; we focus on the total treatment effect (TTE), defined as the difference between the average outcome if all individuals were treated and the average outcome if nobody were treated.
\[\text{TTE} := \tfrac{1}{n} \textstyle\sum_{i \in [n]} (Y_i({\bf 1}) - Y_i({\bf 0})),\]
where ${\bf 1}$ denotes the vector of all ones and ${\bf 0}$ denotes the vector of all zeros. Under heterogeneous additive network effects, the TTE estimand takes value
\[\TTE = \tfrac{1}{n}\left(\textstyle\sum_i \beta_i + \textstyle\sum_{ki} \gamma_{ki}\right).\]

The total treatment effect is particularly relevant in scenarios in which a decision maker can run a randomized experiment with a limited treatment budget, and desires to use the outcome of the experiment to determine whether to fully adopt the new treatment, or to stay with the status quo. The challenge is that the decision maker would like to estimate the outcomes under the all ones treatment vector, but due to a limited budget, it can only observe outcomes under a limited treatment budget. The total treatment effect can be decomposed into a direct effect, a network interference effect, and an interaction effect \cite{SussmanAiroldi17}. The direct effect captures the change in outcomes of an average individual due to itself being treated.
The network interference effect captures the change in outcomes of an average individual due to the network (excluding itself) being treated.
The interaction effect is nonzero in scenarios when the effect of interference on an individual may depend on whether the individual is treated or not. 

Some previous work has focused on estimating the direct effect \cite{BasseAiroldi15, JagadeesanPillaiVolfovsky17, SavjeAronowHudgens17, SussmanAiroldi17, leung2019causal, ma2021causal} or testing for the presence of network interference \cite{Aronow12, BowersFredricksonPanagopoulos12,AtheyEcklesImbens17,PougetAbadieSaveskiSaintJacquesDuanXuGhoshAiroldi17,saveski2017detecting}; these methods do not produce an estimate of the total treatment effect. The techniques for hypothesis testing in \cite{Aronow12, AtheyEcklesImbens17} do not immediately extend to estimation as they are based on randomization inference with a fixed network size and studied testing sharp null hypotheses.
While this paper focuses on estimating the total treatment effect, we show results for the direct treatment effect and the network interference effect in the supplementary material.

\subsection{Class of Estimators}

A primary question of this work is to understand whether one can estimate total treatment effect in the presence of network interference, particularly when the network structure is unknown and costly to estimate, which is often the case in many real world applications. As a result, we consider the following class of individually weighted linear estimators, which have the form
\[ \est(\bw,\bv) = \textstyle\sum_{i \in [n]} \left(w_i z_i + v_i (1 - z_i)\right) Y_i(\bz), \]
where the weights $\bw = (w_1, w_2, \dots w_n)$ and $\bv = (v_1, v_2, \dots v_n)$ are deterministic and not a function of the treatment $\bz$. Most notably, the weight $w_i$ or $v_i$ are only selected based on whether individual $i$ is treated or not, and does not depend on the treatment configuration of its neighbors. 

The focus on linear estimators is not restrictive, as the majority of all estimators proposed in the literature are indeed linear in the measured outcomes. However, the limitation that the weights that multiply an individual's outcome $Y_i(\bz)$ only depends on that individual's treatment $z_i$ is a significant restriction, and arises from the limitation that we do not have knowledge of the network. In contrast, the Horvitz-Thompson estimator under general neighborhood interference is a linear estimator where the weight that multiplies $Y_i(\bz)$ is a function of the treatments of all neighbors of $i$ in addition to $i$ itself; this necessarily requires knowledge of the neighborhood \cite{UganderKarrerBackstromKleinberg13}. All estimators previously studied in the network interference literature require some knowledge of the network.

Under the simplifying SUTVA condition, many classic estimators are in fact individually weighted linear estimators.
For example the Horvitz-Thompson estimator under SUTVA sets $w_i = 1/(n\E[z_i])$ and $v_i = 1/(n\E[1-z_i])$. The difference in means estimator sets $w_i = 1/\sum_{j \in n} z_j$ and $v_i = 1/\sum_{j \in n} (1-z_j)$, which are deterministic for randomized designs in which the total number of individuals under treatment and control are fixed.

\subsection{Discussion of Model Assumptions}

Our model assumes a finite population of size $n$ with arbitrary values for $\alpha_i, \beta_i, \gamma_{ki}$. Since the model allows for any abitrary values, it also can capture a model in which these unknown parameters are generated from an underlying stochastic process, with respect to which the parameters across individuals could be correlated. For example, in an application such as estimating the efficacy of a vaccine, an underlying network mediates the spread of the epidemic under both control and treatment scenarios. If $\{s_i\}_{i \in [n]}$ describes the initial seeds of an infection, and $T_{ji}$ describes the accumulated transmission (potentially over multiple hopes of the network) from a seed $j$ to individual $i$, then the baseline outcomes would take the form of $\alpha_i = \sum_{j \in [n]} T_{ji} s_j$. The baseline parameters $\alpha_i$ are thus correlated with respect to the shared dependence on the random initial seeds $s_j$, but as the baseline parameters are assumed to be indepedent from the treatment assignments, our results will still hold. In particular, our analysis considers the randomness of the outcomes with respect to the treatment vector $\bz$ conditioned on the realized baseline parameters $\alpha_i$, treating them as constant.

This model can also capture network interference that arises from spillover, peer effects, and contagion. Spillover refers to the interference that arises from individual $j$'s treatment affecting individual $i$'s outcome. Typically the spillover effect is assumed to be mediated by the network, such that $\gamma_{ji}$ is non-zero only if $(j,i)$ is an edge in the network. By relaxing constraints on the sparsity of the $\gamma_{ki}$ parameters, the heterogeneous additive network effects assumption can also capture long-range spillover effects mediated by multi-hop paths in the network. 

Contagion or peer effects refer to interference that arises from individual $j$'s outcome affecting individual $i$'s outcome. When the contagion effect is linear, then this translates into long-range interference effects over multi-hop paths in the network. For example, consider a path $\ell \to k \to j \to i$ in the network. Under linear contagion, individual $\ell$'s treatment affects individual $\ell$'s outcome, which subsequently affects individual $k$'s outcome, which subsequently affects individual $j$'s outcome, which then affects individual $i$'s outcome, as described by
\[Y_i({\bf z}) = a_i + b_i z_i + \textstyle\sum_{k \in [n]} c_{ki} Y_k({\bf z}).\]
The potential outcomes model can be derived by solving a system of linear equations for the outcomes vector given an assigned treatment vector, which results in the following potential outcomes model
\[{\bf Y}({\bz}) = (I - C)^{-1} {\bf a} + \textstyle\sum_{t=0}^{\infty} C^t \cdot \text{diag}({\bf b}) \cdot {\bf z},\]
where $C$ is a matrix with the $(k,i)$-th entry equal to $c_{ki}$, $\text{diag}({\bf b})$ is a diagonal matrix with diagonal entries taking values from ${\bf b}$, and ${\bf a}$ and ${\bf b}$ are vectors corresponding to the parameters $\{a_i\}_{i\in[n]}$ and $\{b_i\}_{i\in[n]}$.
Written in this form, we can verify that the potential outcomes could be described via a heterogeneous additive network effects model with dense network effects, as is also observed in \cite{EcklesKarrerUgander17,hu2022average}.

As there are more unknown model parameters than measurements, we cannot hope to identify the model via regression, and thus a randomized experimental design will be critical to any solution. Previous attempts at causal inference under this model either involve complicated network dependent randomized designs \cite{EcklesKarrerUgander17}, incur potentially high network-dependent biases \cite{EcklesKarrerUgander17}, or impose Bayesian priors on the unknown parameters that reduce the statistical estimation task to again estimating a model with a fixed number of parameters \cite{ToulisKao13,BasseAiroldi15}. 

Scenarios that violate linearity include when network effects saturate after a certain number of neighbors are treated, are sublinear in the number of treated neighbors, or are only present after a minimum number of neighbors are treated. Linearity is naturally violated if the measured outcome variable is binary-valued. As a result, our model is more suited to settings where $Y_i$ takes a spectrum of values, such as representing the viral load an individual has accumulated rather than the individual's binary infection status.

\section{Alternate Approaches in the Literature}

A critical challenge for estimating the total treatment effect is that we only observe $\{Y_i(\bz)\}_{i \in [n]}$ for a single fixed treatment vector $\bz$, which is not ${\bf 1}$ or ${\bf 0}$. As a result, we may not observe any of the terms in the expression of interest. Under a fully general arbitrary interference model, it has been repeatedly shown that it is impossible to estimate any desired causal estimands as the model is not fully identifiable \cite{Manski13, AronowSamii17, BasseAiroldi17, karwa2018systematic}. As a result, there have been many proposed models that impose assumptions on exposure functions \cite{Manski13, AronowSamii17, viviano2020experimental, auerbach2021local, li2021causal}, interference neighborhoods \cite{UganderKarrerBackstromKleinberg13, bargagli2020heterogeneous, SussmanAiroldi17, pmlr-v115-bhattacharya20a}, parametric structure \cite{ToulisKao13, BasseAiroldi15, cai2015social, GuiXuBhasinHan15,EcklesKarrerUgander17}, or a combination of these. Each of these assumptions leads to a different solution concept. The art in choosing a good model is balancing the tension between strong assumptions which facilitate simple solutions, and weak assumptions which can more flexibly encompass real world applications. Furthermore, studies show that one must exercise caution in choosing model assumptions, as the results may be sensitive to model misspecification \cite{AronowSamii17,karwa2018systematic}. 

All previously proposed approaches critically rely on using knowledge of a network mediating the interference effects, which is often not available in practice. We highlight a few of the most common models to highlight the strengths and weaknesses of each approach. In complement to the below works, there are ongoing empirical studies assessing the performance difference between an experiment design that leverages the network implicitly, as opposed to a method that measures the network and leverages the measured network \cite{Shakyae012996}.

\subsection{Partial Interference} 
Partial interference assumes that the population can be partitioned into disjoint groups, such that all network interference effects can only occur within but not across the pre-specified groups \cite{Sobel06,Rosenbaum07,HudgensHalloran08,TchetgenVanderWeele12,LiuHudgens14,VanderweeleTchetgenHalloran14,pmlr-v115-bhattacharya20a, auerbach2021local}. Specifically, the outcome of individual $i$ can only depend on the treatment of others in the same group as individual $i$, and is independent from the treatment assignments of individuals in other groups. Under this assumption, we can randomize treatments over the groups jointly so that all individuals in each group are either assigned jointly to treatment or control. As a result, $Y_i(\bz) = Y_i({\bf 1})$ for all $i$ such that $z_i = 1$, and $Y_i(\bz) = Y_i({\bf 0})$ for all $i$ such that $z_i = 0$. 
Unfortunately, this approach does not apply when the network could be highly connected, limiting its use in practice. The bias of standard estimators will scale with the number of edges across clusters, leading to proposed cluster randomized designs that randomize over clusters that are constructed to minimize edges between clusters \cite{GuiXuBhasinHan15, EcklesKarrerUgander17}.

\subsection{Neighborhood Interference}
Under the neighborhood interference assumption, $Y_i(\bz) = Y_i({\bf 1})$ for any individual in the treatment group whose neighbors are also all in the treatment group; we denote this set of individuals $\cS_1(\bz)$. Similarly, $Y_i(\bz) = Y_i({\bf 0})$ for any individual in the control group whose neighbors are also all in the control group; we denote this set $\cS_0(\bz)$. Without imposing any further assumptions, a natural estimate for the total treatment effect is the difference in average outcomes between groups $\cS_1(\bz)$ and $\cS_0(\bz)$, or an inverse probability weighted estimator when the probability of being in group $\cS_1(\bz)$ or $\cS_0(\bz)$ may vary across individuals \cite{AronowSamii17}. Without further structure on the interference, one cannot use measurements from individuals not in either sets $\cS_1(\bz)$ or $\cS_0(\bz)$, as the relationship between $Y_i(\bz)$ and $Y_i({\bf 1})$ or $Y_i({\bf 0})$ is unknown.

Under na\"ive randomized designs such as a Bernoulli design that assigns each individual independently to treatment with probability $p$ or control with probability $1-p$, the variance of the inverse probability weighted estimator 
will go to infinity with $n$ for well-connected networks such as the Erdos-Renyi graph with average degree larger than 
$\sqrt{n}$ \cite{BasseAiroldi17}; this results from the fact that with high probability the sizes of sets $\cS_1(\bz)$ and $\cS_0(\bz)$ will be small for highly connected networks.
As a result, \cite{UganderKarrerBackstromKleinberg13} proposes a graph cluster randomized design that aims to jointly assign individuals and their neighbors to treatment or control in order to minimize variance. Unfortunately this requires detailed knowledge of the network, and constructing optimal clusters can be computationally expensive for non-trivial well-connected networks. As a result, this approach has not translated into practical solutions.

An alternate approach suggested by \cite{hu2022average,li2022random} is to instead consider weaker estimands, which essentially capture the marginal treatment effect of perturbing a status quo treatment assignment distribution. In order to show a central limit theorem styled result, \cite{li2022random} imposes a generative distribution over the network structure itself, and considers how to exploit the regularity in the network that arises from low rank structure. While this enables one to consider networks of increasing size, it may only be plausible in applications in which one can reasonably model the finite network as being sampled from a known generative model.

\subsection{Linear Model with fixed number of  parameters} 

While the above models impose network-based conditions on the interference, an alternate approach is to impose parametric structure on the form of the potential outcomes. The most common assumption is that the potential outcomes are linear with respect to a specified statistic of the local neighborhood \cite{ToulisKao13,GuiXuBhasinHan15,BasseAiroldi15,cai2015social,parker2016optimal,chin2019regression}. For example, \cite{cai2015social} assumes the outcome is linear in the fraction of treated neighbors, such that
\[Y_i(\bz) = \alpha + \beta z_i + \gamma \Big(\frac{\sum_{k \in [n]} A_{ki} z_k}{\sum_{k \in [n]} A_{ki}}\Big),\]
where $A$ is a known adjacency matrix representing edges in the network. Similarly, \cite{parker2016optimal} assumes linearity with respect to the absolute number of treated neighbors. Threshold models also can be expressed with a linear model using indicator statistics. For example \cite{GuiXuBhasinHan15} assumes that network effects arise when at least $\theta$ neighbors are treated, where $\theta$ is assumed to be known,
\[Y_i(\bz) = \alpha + \beta z_i + \gamma ~\Ind\big(\textstyle\sum_{k \in [n]} A_{ki} z_k \geq \theta\big).\]
One can extend these models to incorporate covariate types, such that the total number of unknown parameters is three times the number of different covariate types, assuming that each covariate type is associated to a set of parameters $\alpha,\beta,\gamma$.

What is characteristic of this approach is that the assumptions reduce the number of unknown parameters in the potential outcomes models to a fixed dimension that does not grow with the population size, reducing the inference task to linear regression. As a result, the natural solution is to use a least squares estimate, shifting the focus to constructing randomized designs that minimize the variance of the estimate. 
A limitation of this approach is that it requires the correct choice of the the statistic governing the linearity, and it requires precise knowledge of the network structure to compute these neighborhood statistics. Furthermore, it assumes knowledge of the relevant covariate types that differentiate individual responses, or otherwise assumes homogeneity in the effects. 

\begin{figure}
	\centering
	\includegraphics[width=.5\linewidth]{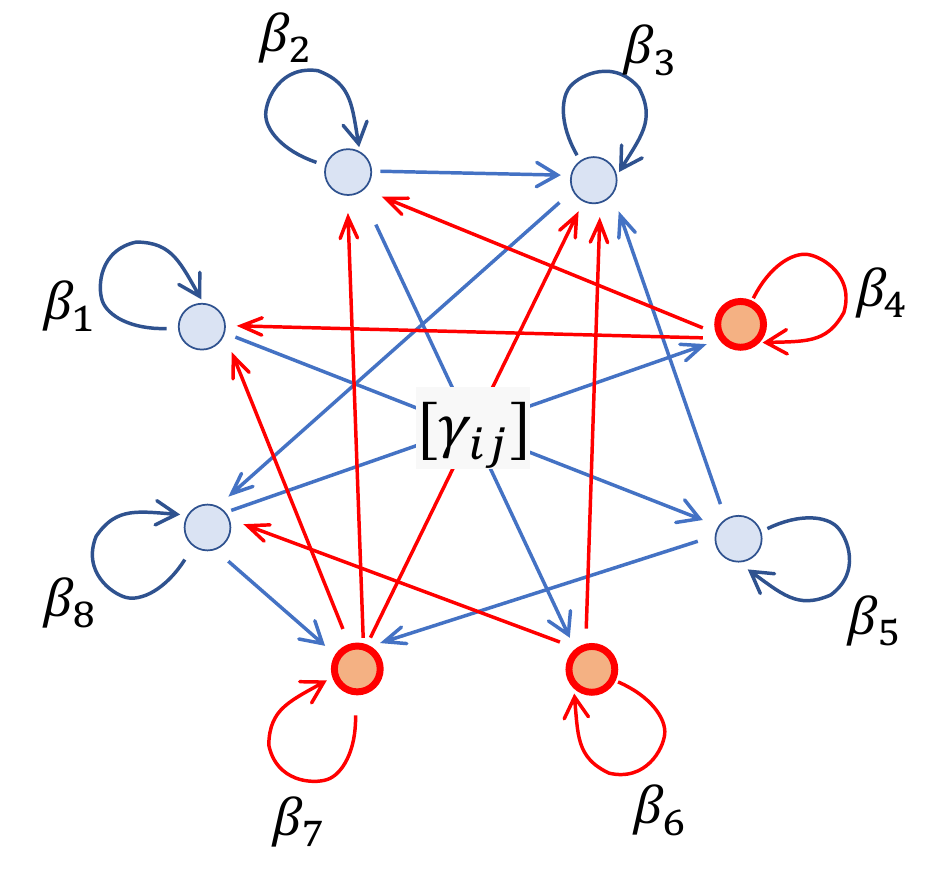}
	\caption{Depiction of model under heterogeneous additive network effects: Vertex weights correspond to baseline outcomes. When individual $i$ is treated (depicted in red), its outgoing edges are activated. Individual $j$'s outcome $Y_j(\bz)$ is the sum of its own baseline plus any incoming activated edges (depicted in red).}
	\label{fig:model}
\end{figure}

\section{Summary of Our Results}

Our results focus on estimating the total treatment effect without knowledge of the network under a heterogeneous additive network effects model, which is significantly more expressive than parametric model classes, where the number of parameters does not scale with the population size. Under our model, the total treatment effect scales linearly in the fraction of treated individuals; our approach exploits this linearity for a simple and efficient solution. Our result is the first to propose a solution without requiring any knowledge of the network. We believe this combination of a practical solution with a flexible model positions our results to have impact in the broader scientific community.

The primary research question is: {\em Does there exist a simple and efficient solution for estimating total treatment effect in the presence of network effects without critically relying on knowledge of the network structure or restrictive network properties?} While this has previously remained elusive, our results provide a clear and simple answer, including both a negative scenario in which there can be no simple solution, and a positive scenario in which we outline a simple efficient solution which can easily translate into practice.  

First, we show that in the presence of additive network effects, any individually weighted linear estimator for the total treatment effect (TTE) is necessarily biased unless the network can be perfectly partitioned into small disjoint subsets with no interfering edges. Furthermore, this bias can be large depending on the relative magnitude of the network effects. This negative result suggests that the linearity arising from the additive network effects model is not sufficient in itself to admit simple estimators which do not utilize knowledge of network structure. The primary reason is that it is difficult for simple estimators to distinguish between the response due to baseline values $\{\alpha_i\}_{i \in [n]}$ as opposed to network effects $\{\gamma_{ji}\}_{(j,i) \in \cE}$. 

Second, we consider the scenario when we have access to an estimate of the average individual baselines; in practice this could be constructed from historical data or pilot studies. Given baseline estimates, we propose a simple estimator for the total treatment effect which computes the average outcomes amongst the entire population after applying the treatment vector, scales the average outcome by the size of the total population divided by the number of treated individuals, and then subtracts the average baseline estimate. This estimator is unbiased for any randomized design in which the marginal probability of an individual being treated is equal amongst different individuals in the population, an easy condition to satisfy as it only involves matching individual marginal treatment probabilities. This estimator is extremely easy to compute, and neither the randomized designs nor the estimate itself require knowledge of the underlying network, which is often not available in practice. 

Third, we show that our proposed approach has low variance under a simple completely randomized design. In particular, the estimator is consistent as long as the fraction of treated individuals is asymptotically larger than $d^2_{\max}/n$, where $d_{\max}$ is the maximum out-degree of any individual in the network, i.e. the maximum influence of any individual in the network. Furthermore, we provide analytical expression for the variance of our estimator under commonly used randomized designs, including completely randomized and cluster randomized design, as well as uniform and varying saturation designs. These variance expressions provide insight for designing randomized designs that minimize variance by matching individuals based on estimated network influence.

\section{Additivity is not sufficient for unbiased estimators}

Assuming heterogeneous additive network effects implies that the total treatment effect scales linearly in the number of treated individuals. A natural question is whether this additive model is sufficient to admit simple unbiased estimators for the total treatment effect or not. In this section we provide some results in the negative by considering the restricted class of individually weighted linear estimators. 

\begin{theorem} \label{thm:TTE}
	Under heterogeneous additive network effects, any unbiased individually weighted linear estimator for total treatment effect must have the form 
	\[ \widehat{\TTE} = \frac{1}{n}\sum_{i \in [n]} \left(\frac{z_i}{\E[z_i]} - \frac{1 - z_i}{\E[1-z_i]}\right) Y_i(\bz), \]
	and the randomized design must satisfy $\Prob(z_k = z_i) = 1$ for all $(k,i) \in \cE$.
	As a result, there does not exist an unbiased individually weighted linear estimator for the total treatment effect if the network is fully connected.
\end{theorem}

\begin{proof}
	Under the heterogeneous additive network effects model, an individually weighted linear estimator takes the value
	\[\est(\bw,\bv) = \textstyle\sum_{i \in [n]} (w_i z_i + v_i (1-z_i)) \alpha_i + \textstyle\sum_{i \in [n]} w_i z_i \beta_i \]
	\[\qquad+ \textstyle\sum_{(k,i) \in \cE} (w_i z_i + v_i (1-z_i)) z_k \gamma_{ki}.\]
	This is an unbiased estimator for the total treatment effect only if $\E[\est(\bw,\bv)] = \frac{1}{n}\sum_{i \in [n]} \beta_i + \frac{1}{n} \sum_{(k,i) \in \cE} \gamma_{ki}$ is satisfied for any configuration of $\{\alpha_i\}_{i \in [n]}, \{\beta_i\}_{i \in [n]}$, and $\{\gamma_{ki}\}_{(k,i) \in \cE}$. This requirement results in the following $2n + |\cE|$ constraints, which arise from matching coefficients for each of the parameters,
	\begin{itemize}
		\item $\alpha$: for all $i \in [n]$, $w_i \E[z_i] + v_i \E[1 - z_i]= 0$,
		\item $\beta$: for all $i \in [n]$, $w_i \E[z_i] = \frac{1}{n}$,
		\item $\gamma$: for all $(k,i) \in \cE$, $w_i \E[z_i z_k] + v_i \E[(1-z_i) z_k] = \frac{1}{n}$.
	\end{itemize}
	The second set of constraints for the direct treatment effects imply that the weights are $w_i = 1/(n\E[z_i])$. As a result, combining this with the first set of constraints for the baselines imply that $v_i = - 1/(n\E[1 - z_i])$. After fixing the values of all the weights $\bw, \bv$, the third set of constraints become difficult to satisfy. We can rewrite the third set of constraints as 
	\[w_i \E[z_i] \Prob(z_k = 1 | z_i = 1) + v_i \E[1-z_i] \Prob(z_k = 1 | z_i = 0) = \frac{1}{n},\]
	for all $(k,i) \in \cE$. Most notably, it is a linear combination of the two terms that show up in the first and second set of constraints, multiplied by probabilities that must be in [0,1]. By plugging in the values for $w_i$ and $v_i$ that arise from the first two constraints, it follows that the third constraint is only satisfied when $\Prob(z_k = 1 | z_i = 1) = 1$ and $\Prob(z_k = 1 | z_i = 0) = 0$ for all $(k,i) \in \cE$. This is equivalent to requiring that the randomized design always assigns connected individuals to the same treatment, i.e. $\Prob(z_k = z_i) = 1$ for all $(k,i) \in \cE$.
\end{proof}

The constraint on the randomized design implies that every pair of connected individuals in the population must be either both treated or both in control. This restricts the valid randomized designs to a cluster randomized design where the clusters are defined by the connected components of the graph. Theorem \ref{thm:TTE} highlights that the imposed structure from heterogeneous additive network effects is insufficient to remove the complex dependence on the network. Even under linearity, we still need to deal with either imposing strong assumptions on the connectivity structure of the network, or we will need to use more complex estimators that utilize knowledge of the network, bringing us back to the same challenges present in the fully general model.

When the conditions for unbiasedness are not satisfied, the bias of the above simple estimator will scale with the average network effect across the edges between the treated and control groups, given by the expression 
\[\E[\widehat{\TTE}] - \text{TTE} = \tfrac{1}{n}\textstyle\sum_{(k,i) \in \cE} \left(\tfrac{\Cov[z_i, z_k]}{\Var[z_i]} - 1\right) \gamma_{ki}.\]
If the randomized design produces high correlation in treatments across pairs of connected individuals in the network, then $\E[\widehat{\TTE}]$ is close in expectation to the total treatment effect. If the design enforces independence of treatments across edges in the network, then $\E[\widehat{\TTE}]$ only captures the direct treatment effects and not the network effects.

The restrictive unbiasedness conditions result from the fact that it is difficult to set the coefficient for the baseline parameters to 0 while maintaining that the coefficients on the network effects are $1/n$, as the expressions for both are very similar. Essentially, it is difficult for the model to distinguish between the effects arising from individual baselines as opposed to the ambient network effects from treated neighbors. Given this insight, in the next section we consider the scenario where we have access to estimates of the average individual baselines.

\section{Simple unbiased estimator given baseline estimates without any knowledge of the network}

In practice there are many applications in which we do have access to additional information from historical data or pilot studies that could be used to construct estimates of the average baseline $\frac{1}{n} \sum_{i \in [n]} Y_i({\bf 0}) = \frac{1}{n} \sum_{i \in [n]} \alpha_i$. For example, a social media platform such as LinkedIn is constantly monitoring the engagement level of its users, such that it always has access to the current status quo baselines at an individual level before deploying randomized trials for a newly proposed feature. 
Even when historical data may not be available, it is typically easy to conduct small scale surveys to estimate the baseline outcome levels before beginning the randomized experiment. The data must be collected before the experiment begins such that no one has yet received the treatment. Under our heterogenous additive network effects assumption, the measurements collected before the experiment will accurately reflect the baseline with no network effects due to treatment; this additional data can then be used to significantly simplify the estimation of causal effects. 

Let's first assume that we have access to the full individual baselines; it follows naturally to then subtract the baseline $\alpha_i$ from the measurement $Y_i(\bz)$ to remove all contributions of the baseline effects from the linear estimator, resulting in 
\[ \est_{-\alpha}(\bw,\bv) = \textstyle\sum_{i \in [n]} \left(w_i z_i + v_i (1 - z_i)\right) (Y_i(\bz) - \alpha_i).\]
To characterize conditions for unbiased linear estimators, we use the same approach of equating the coefficients of the direct effects and the network effects between the expected value of this estimator and the total treatment effect. Subtracting out the baselines removes the set of constraints for unbiasedness associated to the baseline parameters, leaving us with $n + |\cE|$ constraints. While this is still significantly more than the number of measurements, it turns out that there are still many reasonable randomized designs under which we are able to satisfy these constraints. Theorem \ref{thm:TTE_2} presents sufficient conditions for unbiased linear estimators for total treatment effect given baseline estimates.
\begin{theorem} \label{thm:TTE_2}
	Under heterogeneous additive network effects, for any randomized design such that $\frac{\Prob(z_k = 0 ~|~ z_i = 1)}{\Prob(z_k = 1 ~|~ z_i = 0) } = \rho_i$ for all $(k,i) \in \cE$ for some values of $\{\rho_i\}_{i \in [n]}$, the following estimator 
	\[ \widehat{\TTE}_{-\alpha} = \frac{1}{n}\sum_{i \in [n]} \left(\frac{z_i}{\E[z_i]} - \frac{(1 - z_i)\rho_i}{\E[1-z_i]}\right) (Y_i(z) - \alpha_i), \]
	produces an unbiased estimate for the total treatment effect.
\end{theorem}
The condition on the randomization is equivalent to imposing that 
\[\frac{\E[z_i(1-z_k)]}{\E[(1-z_i)z_k]} = \frac{\E[z_i(1-z_j)]}{\E[(1-z_i)z_j]}\]
for all triplets $(i,j,k)$ such that $(k,i) \in \cE$ and $(j,i) \in \cE$. Essentially this boils down to symmetry conditions on the second moments of the treatment vector across edges in the network. Such a symmetry condition would be satisfied by ensuring that for all $i$, the neighbors that influence $i$ are treated equally in the distribution of the assigned treatments. The ratio above can be expanded to 
\[\frac{\E[z_i(1-z_k)]}{\E[(1-z_i)z_k]} = \frac{\E[z_i] - \E[z_i z_k]}{\E[z_k] - \E[z_i z_k]},\]
which is equal to 1 if $\E[z_i] = \E[z_k]$. As a result, a sufficient condition to satisfy the required symmetry is to impose that the marginals are equal across edges, i.e. $\E[z_i] = \E[z_k]$ for all $(k,i) \in \cE$. This is an easy condition to satisfy, and leads to a simplified result as stated below.
\begin{corollary} \label{thm:TTE_corr}
	For any randomized design such that $\E[z_i] = \E[z_k]$ for all $(k,i) \in \cE$, the following simple estimator
	\[\widehat{\TTE}_{-\alpha}  = \frac{1}{n}\sum_{i \in [n]} \frac{Y_i(\bz) -  \alpha_i}{\E[z_i]}\]
	produces an unbiased estimate for the total treatment effect under heterogeneous additive network effects.
	When  $\E[z_i] =p$ for all $i \in [n]$, the estimator further simplifies to
	\[\widehat{\TTE}_{-\alpha}  = \frac{1}{p} \Big(\frac{1}{n}\sum_{i \in [n]} Y_i(\bz) - \frac{1}{n} \sum_{i \in [n]} \alpha_i\Big).\]
\end{corollary}
It may seem that we need knowledge of the network in order to choose a distribution satisfying the symmetry conditions, as they are defined with respect to constraints over the edges, however without knowledge of the network, one could use a distribution with uniform marginal treatment probabilities as it satisfies the required conditions for even the complete graph, which is most restrictive.
When the marginal treatment probability is equal for all individuals, i.e. $\E[z_i] =p$ for all $i \in [n]$, the resulting estimator in fact only needs knowledge of the average population baselines rather than individual baseline parameters. 
While there are a few settings for which individual baseline parameters are observed from historical data, such as experimentation on social media platforms, data of such granularity is not realistic in general. On the other hand, having access an accurate estimate of the average population wide baselines is realistic for a broad variety of applications across public health and social sciences, as the population wide statistic could be estimated from small scale pilot studies. 

Many simple classical randomized designs satisfy the property that all individuals have an equal marginal probability of treatment, in particular, this includes completely randomized design, which assigns a $p$ fraction of individuals uniformly at random from the population to the treatment group. An important property is that {\em neither our estimator nor the appropriate randomized designs need to have knowledge of the underlying network.} In fact, all previously proposed solutions required knowledge of the network either for the randomized design or to compute the estimator. In applications where the network is not fully observed, our proposed estimators will still output in an unbiased estimate for the total treatment effect with simple randomizations that can be implemented without knowledge of the underlying network. This provides positive guarantees for settings in which the randomization may be limited due to regulatory policies or lack of precise network information. 

A critical assumption that our proposed estimator hinges on is the ability to estimate the baseline parameters. While there may not be sufficient information to estimate individual baselines, the estimators presented in Corollary \ref{thm:TTE_corr} only require estimates of the weighted average baseline outcomes, which can be approximated by sampling a small fraction of the population. However, implicit in our assumption is that the average baseline outcomes pre-experiment and post-experiment are the same, which excludes settings in which there are time-related dynamics that significantly change the baselines irregardless of the treatment. As an example, suppose that a pharmaceutical company used reported data from state and national level public health departments to estimate the average baseline rates of contracting COVID-19 before beginning its clinical trials. The assumption that the baseline outcomes remain fairly constant may be violated if the timescale of the trial period is such that the predominate variant of the virus changes in the interim or the baseline natural immunity level of the population changes significantly. 

\begin{figure*}
	\centering
	\includegraphics[width=12 cm]{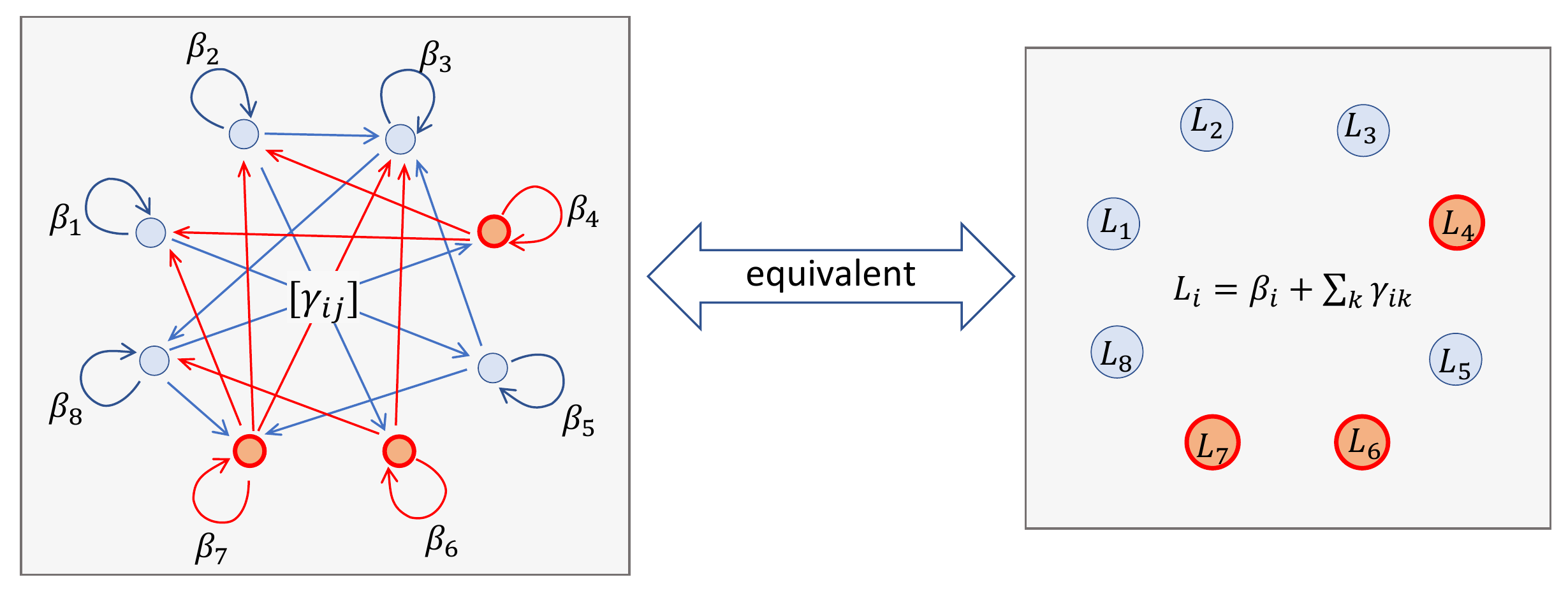}
	\caption{By incorporating baseline estimates, the difficult task of network causal inference reduces to simply estimating a population mean quantity. {\em (Left)} The total treatment effect (TTE) under heterogeneous additive network effects is equivalent to the weighted sum of all the edges divided by the total number of individuals. The vertices corresponding to treated individuals are colored red, and all the outgoing edges from treated individuals are colored red. Given prior knowledge of the population average baseline outcomes, the proposed estimator $\widehat{\TTE}_{-\alpha}$ is equal to the weighted sum of all red edges divided by the number of treated individuals. {\em (Right)} For each vertex $i$, we can sum the weights of outgoing edges into an influence term $L_i$. As a result the task is equivalent to a parallel universe in which each vertex $i$ is associated to an influence $L_i$, and there is no further network interaction. As a result of the proposed experiment, we observe the sum of the influences of all treated individuals. Any randomized design with enough randomness and regularity will be able to guarantee that the average influence of the treated individuals is equal to the population mean influence.} 
	\label{fig:key_insight}
\end{figure*}

\section{Reduction from Network Causal Inference to Estimation of Population Mean}

For the remainder of the paper we will focus on the following estimator introduced in Corollary \ref{thm:TTE_corr},
\[\widehat{\TTE}_{-\alpha}  = \frac{1}{n}{\sum_{i \in [n]} }\frac{Y_i(\bz) -  \alpha_i}{\E[z_i]}.\]
It is easy to verify that the bias of the estimator is given by
\[\E[\widehat{\TTE}_{-\alpha}] - \text{TTE} = \tfrac{1}{n}\textstyle\sum_{(i,k) \in \cE} \left(\tfrac{\E[z_i]}{\E[z_k]} - 1\right) \gamma_{ik},\]
which is zero when the marginal treatment probabilities are equal across individuals. This estimator is particularly simple because the weights are chosen such that $w_i = v_i$, i.e. each outcome is incorporated to the estimator with the same weight regardless of whether it is treated or not. 

We define the ``influence'' of individual $i$ on the estimate as 
\[L_i = \beta_i + \sum_{k \in [n]} \frac{\E[z_i] \gamma_{ik}}{\E[z_k]}.\]
We refer to this as ``influence'' because $L_i$ captures the contribution that individual $i$ has towards the estimate $\widehat{\TTE}_{-\alpha}$ when $i$ is treated, including the interference effect it has on other individuals as well as the the direct effect it has on itself. $L_i$ does not depend on the realization of the treatment vector.

Under heterogeneous additive network effects,
\[\widehat{\TTE}_{-\alpha} = \frac{1}{n}\sum_{i \in [n]} \Big(\frac{\beta_i}{\E[z_i]} + \sum_{k \in [n]} \frac{\gamma_{ik}}{\E[z_k]}\Big) z_i =: \frac{1}{n}\sum_{i \in [n]} \frac{L_i z_i}{\E[z_i]}.\]
Written in this form, it is clear that the $\widehat{\TTE}_{-\alpha}$ is simply an inverse propensity weighted estimator, although the terms $L_i$ are not actually observed. As a result, analytical expressions of the variance of our estimate will follow from direct calculations over the inverse propensity weighted estimates. 

Furthermore, under the sufficient conditions for unbiasedness when $\E[z_i] = \E[z_k]$ for all $(i,k) \in \cE$, the total treatment effect is equal to the population mean of the influence terms,
\[TTE = \frac{1}{n} \sum_{i \in [n]} L_i.\]
As a result, under the mild assumption of having access to baseline estimates, we have reduced the complex task of network causal inference to a simple task of estimating a population mean of the influence via the sample average of the influence of treated individuals. This removes all complexity of network interference as the influence terms do not interact.
Although our randomized designs and estimators do not require any knowledge of the network, the distribution of the influence will depend on the network, and will subsequently affect the variance of our estimator. 

\section{Variance of Proposed Estimator} \label{sec:variance}

As there are still many different randomized treatment designs one could use, we compare the variance of our proposed estimator under some commonly used randomized designs. Let $p$ denote the treatment budget, such that the size of the treatment group can be at most $p$ fraction of the population.

As our estimator corresponds to approximating the total treatment effect with the average of the influence of the treated individuals, this estimator directly inherits properties from the analysis of a sample average estimator for the population mean of the influence. The variance of the estimator is
\begin{align*}
	\Var[\widehat{\TTE}_{-\alpha}] = \sum_{i,j}\frac{L_i L_j \Cov(z_i, z_j)}{n^2 \E[z_i] \E[z_j]}.
\end{align*}
The randomized design affects the variance through the covariance matrix of the treatment vector.

\subsubsection{Completely Randomized Design (CRD)}
The completely randomized design generates the treatment assignment vector $\bz$ by selecting a subset of $pn$ units to treat uniformly at random out of the size $n$ population. This randomized design is commonly used in the classical setting without network interference. Due to the uniform sampling, $\E[z_i]=p$ for all $i \in [n]$, and
\[\Var[\widehat{\TTE}_{-\alpha}] = \frac{1-p}{p(n-1)} \Big(\frac{1}{n} \sum_{i\in[n]} L_i^2 - \Big(\frac{1}{n}\sum_{i\in[n]} L_i\Big)^2\Big).\]
The inner expression is equal to the population variance of the influence terms $\{L_i\}_{i \in [n]}$, which is bounded by $B^2 d^2_{\max}$, where $d_{\max}$ denotes the maximum out-degree of the network, and $B$ is a bound on the direct effect and network effect parameters. A simple bound on the variance thus scales as $B^2 d^2_{\max}/pn$. As a result, the variance will converge to zero for large $n$ as long as the number of treated individuals $pn$ divided by a constant $B^2 d^2_{\max}$ goes to infinity with large $n$. This is an optimally efficient rate when $B$ and $d_{\max}$ are constants, as this only requires the number of treated individuals to be growing larger than a constant as $n$ grows. In fact, even optimal estimators for causal effects under SUTVA, without network interference, results in a variance scaling as $1/pn$. This means that given the mild assumption of having access to baseline estimates, our approach under fully heterogeneous network effects attains a simple, unbiased, and optimally efficient estimate for the total treatment effect, under the simplest randomized design.

\subsubsection{Cluster Randomized Design (cluster RD)}
The cluster randomized design partitions the population into clusters, and all individuals in each cluster are either jointly placed in the treatment group or the control group. This is also referred to in the literature as block randomized design, where blocks refer to clusters. The treatment assignment vector is generated by selecting a subset of $pT$ clusters to treat uniformly at random amongst the $T$ clusters. In contrast to completely randomized design, the treatment of individuals within a cluster are perfectly correlated. This randomized design is commonly used in the network interference setting, where the clusters are additionally constructed to minimize edges across clusters, so that an individual and its local neighbors are jointly assigned to treatment or control as much as possible. In our setting, we do not require such conditions on the construction of the clusters, and thus our clusters may not correspond to tightly connected communities in the network.

The probability an individual is treated is equal to the probability that the cluster it belongs to is treated. Due to the uniform sampling across the clusters, it follows that $\E[z_i]=p$ for all $i \in [n]$. Let $T$ denote the number of clusters, assuming clusters of uniform size $n/T$ for simplicity. Let $\pi:[n] \to [T]$ denote the mapping that assigns individuals to clusters. Let $L'_{\tau}$ denote the average value of the influence terms within cluster $\tau$,
$L'_{\tau} = \frac{T}{n}\textstyle\sum_{i: \pi(i) = \tau} L_i$.
The variance of our estimator under this randomized design is given by
\[\Var[\widehat{\TTE}_{-\alpha} ] 
= \frac{1-p}{p(T-1)} \Big(\frac{1}{T} \sum_{\tau \in [T]} {L'_{\tau}}^2 - \Big(\frac{1}{T}\sum_{\tau \in [T]} L'_{\tau} \Big)^2\Big).\]
Observe that the expression is very similar to the variance under the completely randomized design, except we are randomizing over clusters rather than individuals. The inner expression is equal to the variance across clusters of $L'_{\tau}$, which is the average influence of individuals in cluster $\tau$, bounded by $B^2 (\max_{\tau \in [T]} \bar{d}_{\tau})^2$, where $\bar{d}_{\tau}$ denotes the average out-degree of individuals in cluster $\tau$. A bound on the variance thus scales as $B^2  (\max_{\tau \in [T]} \bar{d}_{\tau})^2 / p T$.

If $d_{\max}$ is constant, and if $T$ is asymptotically smaller than $n$, i.e. the size of each cluster is not constant with respect to the population size $n$, then the variance under cluster randomized design is larger than the variance under completely randomized design. When $d_{\max}$ may be large or even growing with $n$ then the variance might be improved by using cluster randomized design with an optimal choice of clusters. In particular, there would be a tradeoff between the choice of cluster size, and the gain from smoothing out the degree distribution, i.e. the influence, across clusters. In particular, if there is high variation in the influence amongst individuals, then the variance would be minimized by splitting large influence individuals across different clusters, and grouping them with low influence individuals, in order to try to even out the average influence of each cluster. This requires detailed knowledge of the network however, which is often not available.

\subsubsection{Saturation Randomized Design (SRD)}
The saturation randomized design also assumes that the population is partitioned into clusters, but instead each cluster is treated at a specified saturation level, specifying the percentage of individuals in that cluster that are treated. This is also referred to as a cluster stratified randomized design. For a cluster $\tau$, let $p_{\tau}$ denote the fraction of individuals treated in cluster $\tau$, satisfying $\sum_{\tau \in [T]} n_{\tau} p_{\tau} = np$. Let $T$ denote the number of clusters, and let $\pi:[n] \to [T]$ map individuals to clusters. Let $n_{\tau} = |{i: \pi(i) = \tau}|$ denote the size of cluster $\tau$.  Treatment across different clusters are assigned independently. For each cluster $\tau$, a set of $p_{\tau} n_{\tau}$ individuals within the cluster are selected uniformly at random to be treated. Let $V_{\tau}$ denote the variance of the influence terms within cluster $\tau$, 
\[V_{\tau} = \frac{1}{n_{\tau}} \sum_{i: \pi(i) = \tau} L_i^2 - \Big(\frac{1}{n_{\tau}}\sum_{i: \pi(i) = \tau} L_i\Big)^2.\] 
Under uniform saturation, i.e. $p_{\tau} = p$ for all $\tau$, the estimator is unbiased as $\E[z_i] = p$ for all $i$, and the variance is
\[\Var[\widehat{\TTE}_{-\alpha} ] = \frac{1-p}{p n} \sum_{\tau \in [T]} \frac{n_{\tau}^2}{n(n_{\tau}-1)} V_{\tau}.\]
This expression essentially scales as $V_{\text{avg}}/pn$, where $V_{\text{avg}}$ denotes the weighted average across cluster variances $V_{\tau}$. In particular, to minimize the variance, each cluster should be chosen to be as homogeneous as possible, so that there is little variation in the influence parameters within each cluster. If $V_{\text{avg}}$ is significantly smaller than the overall population variance over the influence terms $\{L_i\}_{i \in [n]}$, then the uniform saturation randomized design improves in efficiency upon CRD. 
An extreme special case of this randomized design would be the matched pair randomized design, where the pairs correspond to the clusters, and $p = 1/2$. The pairs are selected to be as similar as possible on known features. 

Constructing such clusters requires additional knowledge of covariates or network structure which is not always available; however, this analysis provides motivation that whenever we do have such auxiliary information at hand, we can only benefit by controlling for the variance which may be related to the auxiliary information. In particular, by grouping similar individuals together, we ensure that the distribution over the auxiliary information in the treated group is as similar as possible to the control group.

For a general choice of saturation levels, the estimator may be biased, and the bias will scale proportionally with the sum of the network effect of edges across clusters with different saturation levels. The variance of the estimator is given by
\[\Var[\widehat{\TTE}_{-\alpha} ] = \sum_{\tau \in [T]} \frac{(1-p_{\tau}) n_{\tau}^2}{p_{\tau} n^2(n_{\tau}-1)} V_{\tau}.\]
While varying the saturation levels would introduce bias into the estimator, it may be able to reduce the variance by allocating a larger treated fraction to clusters that are larger or that have larger within-cluster variance $V_{\tau}$. the reduction in variance would need to be carefully balanced with the introduced bias however, which may be difficult to do since it would require auxiliary information about the cluster variances. 

\section{Conclusion}

Estimating the total treatment effect under network interference is an important yet challenging problem. Previous solutions under general models are often too computationally or statistically costly, or are limited to very simplistic network structures, inhibiting adoption in practice. As a result, many practical solutions consider strong assumptions on the network effects, which end up reducing the estimation task to a simple regression problem. In contrast, we consider the heterogeneous additive network effects assumption, which imposes additive network effects, but allows for full heterogeneity in the edge level network effect parameters. This model is significantly more flexible than the simple linear models used in practice. We analyze the properties of individually weighted linear estimators under our model, and our results directly translate into the following insights that are simple to apply in practice. Given baseline estimates, we show that {\em network causal inference is as easy as estimating a population mean!} Most notably, our solution {\em does not require knowledge of the underlying network}, nor does it critically require strong structural conditions on the network. 

\paragraph{Insight 1: Prior information from historical data or pilot studies is incredibly valuable; without such information, any unbiased estimate must use knowledge of the network.}
We showed that without prior information, there does not exist any unbiased estimate for the total treatment effect which does not use network structure. In particular, we restricted to linear estimators with weights that depend only on whether an individual is treated or not, and not on the treatment of its neighbors. We showed that an unbiased linear estimator only exists if the network can be fully partitioned into disconnected components, such that the randomized design must jointly treat or not treat all individuals in each connected component. 

\paragraph{Insight 2: Use historical data or pilot studies to estimate the population baseline; use the following simple unbiased estimator to approximate the total treatment effect.}
We proposed the following unbiased estimator for any randomized design for which the marginal treatment probability of each individual is $p$,
\[\widehat{\TTE}_{-\alpha}  = \frac{1}{p} \left(\frac{1}{n}\sum_{i \in [n]} Y_i(\bz) - \frac{1}{n} \sum_{i \in [n]} \alpha_i\right).\]

\paragraph{Insight 3: The statistical properties of the estimator depend on the population distribution of individual influence on the total treatment effect.} Under our additive network effects model, our proposed estimator takes the form of 
\[\widehat{\TTE}_{-\alpha}  = \frac{1}{n}\sum_{i \in [n]} \frac{L_i z_i}{\E[z_i]} ~\text{ for }~ L_i = \beta_i + \sum_{k \in [n]} \frac{\E[z_i] \gamma_{ik}}{\E[z_k]},\]
where $L_i$ quantifies individual $i$'s influence on the total treatment effect. This characterization shows that for a sufficiently large population, under simple randomized designs, the distribution of the estimator will also be approximately Gaussian. As a result, variance estimates could be used to design hypothesis tests and compute p-values.

\paragraph{Insight 4: Using completely randomized design (CRD) results in optimally efficient estimation of the total treatment effect when the effect sizes and out-degrees are bounded.} Under CRD, the variance of our estimator is roughly equal to $1/pn$ times the population variance of the influence terms $\{L_i\}_{i \in [n]}$, which is bounded by $B^2 d_{\max}^2$ when the causal effect parameters are bounded by $B$ and the out-degree of each individual is bounded by $d_{\max}$. As a result, the estimator is consistent as long as the number of treated individuals $pn$ is larger than a constant with respect to the population size $n$.

\paragraph{Insight 5: Utilize any auxiliary information about network structure or covariates of individuals in the population to control for variance that may arise from heterogeneity amongst individuals.} The variance of the estimator is minimized by constructing a randomized design under which the distribution of the influence terms within treated and control are as similar as possible. The influence of individual $i$, denoted by $L_i$, is a function of the causal effect parameters $\beta_i, \gamma_{ik}$, but also depends on the network via its local neighborhood structure. If the causal effect parameters or the network structure were related to observed covariantes, we could reduce the variance of the estimator by using a uniform saturation randomized design, where we group together individuals that are similar with respect to observed covariates and local neighborhood structure in the network.

\medskip
There are many interesting possible extensions for this work, including generalizing the model to add nonlinear terms or relaxing the requirement on having access to estimates of the population baseline, which are being studied by a subset of the authors in forthcoming papers \cite{cortez2022graph, cortez2022exploiting}.

\section{Acknowledgements}

This research was supported in part by Microsoft Research New England. We also gratefully acknowledge funding from the NSF under grants CCF-1948256 and CNS-1955997. Christina Lee Yu is also supported by an Intel Rising Stars Award.
	
\bibliographystyle{plain}
\bibliography{bibliography}

\begin{thebibliography}{10}

\bibitem{Aronow12}
Peter~M. Aronow.
\newblock A general method for detecting interference between units in
  randomized experiments.
\newblock {\em Sociological Methods \& Research}, 41(1):3--16, 2012.

\bibitem{AronowSamii17}
Peter~M Aronow, Cyrus Samii, et~al.
\newblock Estimating average causal effects under general interference, with
  application to a social network experiment.
\newblock {\em The Annals of Applied Statistics}, 11(4):1912--1947, 2017.

\bibitem{AtheyEcklesImbens17}
Susan Athey, Dean Eckles, and Guido~W Imbens.
\newblock Exact p-values for network interference.
\newblock {\em Journal of the American Statistical Association},
  113(521):230--240, 2018.

\bibitem{auerbach2021local}
Eric Auerbach and Max Tabord-Meehan.
\newblock The local approach to causal inference under network interference.
\newblock {\em arXiv preprint arXiv:2105.03810}, 2021.

\bibitem{bargagli2020heterogeneous}
Falco~J Bargagli-Stoffi, Costanza Tort{\`u}, and Laura Forastiere.
\newblock Heterogeneous treatment and spillover effects under clustered network
  interference.
\newblock {\em arXiv preprint arXiv:2008.00707}, 2020.

\bibitem{BasseAiroldi17}
Guillaume~W. Basse and Edoardo~M. Airoldi.
\newblock Limitations of design-based causal inference and a/b testing under
  arbitrary and network interference.
\newblock {\em Sociological Methodology}, 48(1):136--151, 2018.

\bibitem{BasseAiroldi15}
Guillaume~W Basse and Edoardo~M Airoldi.
\newblock Model-assisted design of experiments in the presence of
  network-correlated outcomes.
\newblock {\em Biometrika}, 105(4):849--858, 2018.

\bibitem{pmlr-v115-bhattacharya20a}
Rohit Bhattacharya, Daniel Malinsky, and Ilya Shpitser.
\newblock Causal inference under interference and network uncertainty.
\newblock In Ryan~P. Adams and Vibhav Gogate, editors, {\em Proceedings of The
  35th Uncertainty in Artificial Intelligence Conference}, volume 115 of {\em
  Proceedings of Machine Learning Research}, pages 1028--1038. PMLR, 22--25 Jul
  2020.

\bibitem{BowersFredricksonPanagopoulos12}
Jake Bowers, Mark~M. Fredrickson, and Costas Panagopoulos.
\newblock Reasoning about interference between units: A general framework.
\newblock {\em Political Analysis}, 21(1):97--124, 2013.

\bibitem{cai2015social}
Jing Cai, Alain De~Janvry, and Elisabeth Sadoulet.
\newblock Social networks and the decision to insure.
\newblock {\em American Economic Journal: Applied Economics}, 7(2):81--108,
  2015.

\bibitem{chin2019regression}
Alex Chin.
\newblock Regression adjustments for estimating the global treatment effect in
  experiments with interference.
\newblock {\em Journal of Causal Inference}, 7(2), 2019.

\bibitem{cortez2022exploiting}
Mayleen Cortez, Matthew Eichhorn, and Christina~Lee Yu.
\newblock Exploiting neighborhood interference with low order interactions
  under unit randomized design.
\newblock {\em arXiv preprint arXiv:2208.05553}, 2022.

\bibitem{cortez2022graph}
Mayleen Cortez, Matthew Eichhorn, and Christina~Lee Yu.
\newblock Staggered rollout designs enable causal inference under interference
  without network knowledge.
\newblock {\em arXiv preprint arXiv:2205.14552}, 2022.

\bibitem{cox1958planning}
David~Roxbee Cox.
\newblock {\em Planning of experiments.}
\newblock Wiley, 1958.

\bibitem{EcklesKarrerUgander17}
Dean Eckles, Brian Karrer, and Johan Ugander.
\newblock Design and analysis of experiments in networks: Reducing bias from
  interference.
\newblock {\em Journal of Causal Inference}, 5(1), 2017.

\bibitem{GuiXuBhasinHan15}
Huan Gui, Ya~Xu, Anmol Bhasin, and Jiawei Han.
\newblock Network a/b testing: From sampling to estimation.
\newblock In {\em Proceedings of the 24th International Conference on World
  Wide Web}, pages 399--409. International World Wide Web Conferences Steering
  Committee, 2015.

\bibitem{hu2022average}
Yuchen Hu, Shuangning Li, and Stefan Wager.
\newblock Average direct and indirect causal effects under interference.
\newblock {\em Biometrika}, 2022.

\bibitem{HudgensHalloran08}
Michael~G. Hudgens and M.~Elizabeth Halloran.
\newblock Toward causal inference with interference.
\newblock {\em Journal of the American Statistical Association}, 103:832--842,
  2008.

\bibitem{JagadeesanPillaiVolfovsky17}
Ravi Jagadeesan, Natesh~S Pillai, and Alexander Volfovsky.
\newblock Designs for estimating the treatment effect in networks with
  interference.
\newblock {\em The Annals of Statistics}, 48(2):679--712, 2020.

\bibitem{karwa2018systematic}
Vishesh Karwa and Edoardo~M Airoldi.
\newblock A systematic investigation of classical causal inference strategies
  under mis-specification due to network interference.
\newblock {\em arXiv preprint arXiv:1810.08259}, 2018.

\bibitem{leung2019causal}
Michael~P Leung.
\newblock Causal inference under approximate neighborhood interference.
\newblock {\em Econometrica}, 90(1):267--293, 2022.

\bibitem{li2022random}
Shuangning Li and Stefan Wager.
\newblock Random graph asymptotics for treatment effect estimation under
  network interference.
\newblock {\em The Annals of Statistics}, 50(4):2334--2358, 2022.

\bibitem{li2021causal}
Wenrui Li, Daniel~L Sussman, and Eric~D Kolaczyk.
\newblock Causal inference under network interference with noise.
\newblock {\em arXiv preprint arXiv:2105.04518}, 2021.

\bibitem{LiuHudgens14}
Lan Liu and Michael~G. Hudgens.
\newblock Large sample randomization inference of causal effects in the
  presence of interference.
\newblock {\em Journal of the American Statistical Association},
  109(505):288--301, 2014.
\newblock PMID: 24659836.

\bibitem{ma2021causal}
Yunpu Ma and Volker Tresp.
\newblock Causal inference under networked interference and intervention policy
  enhancement.
\newblock In {\em International Conference on Artificial Intelligence and
  Statistics}, pages 3700--3708. PMLR, 2021.

\bibitem{Manski13}
Charles~F Manski.
\newblock Identification of treatment response with social interactions.
\newblock {\em The Econometrics Journal}, 16(1), 2013.

\bibitem{parker2016optimal}
Ben~M. Parker, Steven~G. Gilmour, and John Schormans.
\newblock Optimal design of experiments on connected units with application to
  social networks.
\newblock {\em Journal of the Royal Statistical Society: Series C (Applied
  Statistics)}, 66(3):455--480, 2017.

\bibitem{PougetAbadieSaveskiSaintJacquesDuanXuGhoshAiroldi17}
Jean Pouget-Abadie, Guillaume Saint-Jacques, Martin Saveski, Weitao Duan,
  S~Ghosh, Y~Xu, and Edoardo~M Airoldi.
\newblock Testing for arbitrary interference on experimentation platforms.
\newblock {\em Biometrika}, 106(4):929--940, 2019.

\bibitem{Rosenbaum07}
Paul~R Rosenbaum.
\newblock Interference between units in randomized experiments.
\newblock {\em Journal of the American Statistical Association},
  102(477):191--200, 2007.

\bibitem{rubin1978bayesian}
Donald~B. Rubin.
\newblock {Bayesian Inference for Causal Effects: The Role of Randomization}.
\newblock {\em The Annals of Statistics}, 6(1):34 -- 58, 1978.

\bibitem{saveski2017detecting}
Martin Saveski, Jean Pouget-Abadie, Guillaume Saint-Jacques, Weitao Duan,
  Souvik Ghosh, Ya~Xu, and Edoardo~M Airoldi.
\newblock Detecting network effects: Randomizing over randomized experiments.
\newblock In {\em Proceedings of the 23rd ACM SIGKDD International Conference
  on Knowledge Discovery and Data Mining}, pages 1027--1035. ACM, 2017.

\bibitem{SavjeAronowHudgens17}
Fredrik S{\"a}vje, Peter~M Aronow, and Michael~G Hudgens.
\newblock Average treatment effects in the presence of unknown interference.
\newblock {\em The Annals of Statistics}, 49(2):673--701, 2021.

\bibitem{Shakyae012996}
Holly~B Shakya, Derek Stafford, D~Alex Hughes, Thomas Keegan, Rennie Negron,
  Jai Broome, Mark McKnight, Liza Nicoll, Jennifer Nelson, Emma Iriarte, Maria
  Ordonez, Edo Airoldi, James~H Fowler, and Nicholas~A Christakis.
\newblock Exploiting social influence to magnify population-level behaviour
  change in maternal and child health: study protocol for a randomised
  controlled trial of network targeting algorithms in rural honduras.
\newblock {\em BMJ Open}, 7(3), 2017.

\bibitem{Sobel06}
Michael~E Sobel.
\newblock What do randomized studies of housing mobility demonstrate?
\newblock {\em Journal of the American Statistical Association},
  101(476):1398--1407, 2006.

\bibitem{SussmanAiroldi17}
Daniel~L Sussman and Edoardo~M Airoldi.
\newblock Elements of estimation theory for causal effects in the presence of
  network interference.
\newblock {\em arXiv preprint arXiv:1702.03578}, 2017.

\bibitem{TchetgenVanderWeele12}
Eric J~Tchetgen Tchetgen and Tyler~J VanderWeele.
\newblock On causal inference in the presence of interference.
\newblock {\em Statistical Methods in Medical Research}, 21(1):55--75, 2012.
\newblock PMID: 21068053.

\bibitem{ToulisKao13}
Panos Toulis and Edward Kao.
\newblock Estimation of causal peer influence effects.
\newblock In {\em International Conference on Machine Learning}, pages
  1489--1497, 2013.

\bibitem{UganderKarrerBackstromKleinberg13}
Johan Ugander, Brian Karrer, Lars Backstrom, and Jon Kleinberg.
\newblock Graph cluster randomization: Network exposure to multiple universes.
\newblock In {\em Proceedings of the 19th ACM SIGKDD international conference
  on Knowledge discovery and data mining}, pages 329--337. ACM, 2013.

\bibitem{VanderweeleTchetgenHalloran14}
Tyler~J. VanderWeele, Eric~J. Tchetgen~Tchetgen, and M.~Elizabeth Halloran.
\newblock Interference and sensitivity analysis.
\newblock {\em Statist. Sci.}, 29(4):687--706, 11 2014.

\bibitem{viviano2020experimental}
Davide Viviano.
\newblock Experimental design under network interference.
\newblock {\em arXiv preprint arXiv:2003.08421}, 2020.

\end{thebibliography}

\appendix

\section{Proof of Results on Total Treatment Effect}





\begin{proof}[Proof of Theorem \ref{thm:TTE_2}]
	After subtracting out the baseline parameters, the constraints for unbiasedness reduce to 
	\begin{itemize}
		\item $\beta$: for all $i \in [n]$, $w_i \E[z_i] = \frac{1}{n}$,
		\item $\gamma$: for all $(k,i) \in \cE$, $w_i \E[z_i z_k] + v_i \E[(1-z_i) z_k] = \frac{1}{n}$.
	\end{itemize}
	Satisfying the constraints arising from $\beta$ and $\gamma$ results in
	\[w_i = \frac{1}{n \E[z_i]} \text{ and } v_i = \frac{\E[z_i(1-z_k)]}{n \E[z_i]\E[(1-z_i)z_k]} \text{ for all } (k,i) \in \cE.\]
	In order to ensure that such a valid $v_i$ exists, we would need that $\frac{\Prob(z_k = 0 ~|~ z_i = 1)}{\Prob(z_k = 1 ~|~ z_i = 0) } = \rho_i$ for all $(k,i) \in \cE$.
	Under this condition,
	\[\frac{\E[z_i(1-z_k)]}{\E[z_i]\E[(1-z_i)z_k]} = 
	\frac{\E[z_i] \Prob(z_k = 0 ~|~ z_i = 1)}{\E[z_i]\E[(1-z_i)] \Prob(z_k = 1 ~|~ z_i = 0)} = \frac{\rho_i}{\E[(1-z_i)]}.\]
\end{proof}


\begin{proof}[Proof of Corollary \ref{thm:TTE_corr}]
	First we can verify that if $\E[z_i] = \E[z_k]$, then 
	\[\frac{\Prob(z_k = 0 ~|~ z_i = 1)}{\Prob(z_k = 1 ~|~ z_i = 0) } =\frac{\E[z_i(1-z_k)]\E[1-z_i]}{\E[z_i]\E[z_k(1-z_i)]} =\frac{(\E[z_i]-\E[z_i z_k])\E[1-z_i]}{\E[z_i](\E[z_k]-\E[z_k z_i])}
	=\frac{\E[1-z_i]}{\E[z_i]} =: \rho_i.\]
	The estimator then results from plugging in the expression for $\rho_i$ into the estimator defined in Theorem \ref{thm:TTE_2}. When $\E[z_i] = p$, rearranging the expression then results in the simplified form  \[\widehat{\TTE}_{-\alpha}  = \frac{1}{p} \left(\frac{1}{n}\sum_{i \in [n]} Y_i(\bz) - \frac{1}{n} \sum_{i \in [n]} \alpha_i\right),\]
	highlighting that the only knowledge of the baseline parameters needed is the population average baseline.
\end{proof}

\section{Estimating the Average Treatment Effect}

The average treatment effect (ATE), also referred to as the direct treatment effect, measures the average difference in outcomes for individuals that is caused only by their own treatments, not including any network effects. It is formally defined as
\[\ATE := \frac{1}{n}\sum_{i \in [n]} (Y_i({\bf e}_i) - Y_i({\bf 0}))
= \frac{1}{n}\sum_{i \in [n]}\beta_i,\]
where ${\bf e}_i$ is the standard basis vector with 1 at component $i$ and zero elsewhere. 

\begin{theorem} \label{thm:ATE}
	under heterogeneous additive network effects, any unbiased individually weighted linear estimator for average treatment effect must have the form
	\[ \widehat{\ATE} = \frac{1}{n}\sum_{i \in [n]} \left(\frac{z_i}{\E[z_i]} - \frac{1 - z_i}{\E[1-z_i]}\right) Y_i(\bz), \]
	and the randomized design must satisfy $z_k \indep z_i$ for all $(k,i) \in \cE$.
\end{theorem}

\begin{proof}
	Under the heterogeneous additive network effects model, an individually weighted linear estimator takes the value
	\[\est(\bw,\bv) = \textstyle\sum_{i \in [n]} (w_i z_i + v_i (1-z_i)) \alpha_i + \textstyle\sum_{i \in [n]} w_i z_i \beta_i \]
	\[\qquad+ \textstyle\sum_{(k,i) \in \cE} (w_i z_i + v_i (1-z_i)) z_k \gamma_{ki}.\]
	This is unbiased for the average treatment effect only if $\E[\est(\bw,\bv)] = \frac{1}{n}\sum_{i \in [n]} \beta_i$ is satisfied for any configuration of $\{\alpha_i\}_{i \in [n]}, \{\beta_i\}_{i \in [n]}$, and $\{\gamma_{ki}\}_{(k,i) \in \cE}$. This requirement results in the following $2n + |\cE|$ constraints, which arise from matching coefficients for each of the parameters,
	\begin{itemize}
		\item $\alpha$: for all $i \in [n]$, $w_i \E[z_i] + v_i \E[1 - z_i]= 0$,
		\item $\beta$: for all $i \in [n]$, $w_i \E[z_i] = \frac{1}{n}$,
		\item $\gamma$: for all $(k,i) \in \cE$, $w_i \E[z_i z_k] + v_i \E[(1-z_i) z_k] = 0$.
	\end{itemize}
	Solving for the weights given the first two constraints results in 
	\[w_i = \frac{1}{n\E[z_i]} \text{ and } v_i = - \frac{1}{n\E[1-z_i]}.\]
	By plugging in these values of $w_i$ and $v_i$ into the third set of constraints arising from $\gamma$, it follows that we must satisfy 
	\[\frac{\E[z_i z_k]}{n\E[z_i]} - \frac{\E[(1-z_i) z_k]}{n\E[1-z_i]} =
	\frac{\Prob(z_k =1 | z_i = 1)}{n} - \frac{\Prob(z_k =1 | z_i = 0)}{n} = 0,\]
	which implies that $\Prob(z_k =1 | z_i = 1) = \Prob(z_k =1 | z_i = 0)$ such that $z_k \indep z_i$ for all $(k,i) \in \cE$.
\end{proof}

This independence constraint on the randomized design is quite restrictive, yet is easily satisfied by simple Bernoulli randomization. Although many other randomizations may not satisfy independence, if they are ``almost'' independent, for example in completely randomized design with a large population, then the bias will still be small.

\begin{corollary}\label{cor:ATE}
	Under heterogeneous additive network effects, the Horvitz-Thompson estimator with Bernoulli randomization is an unbiased estimator for the average treatment effect, even when the interference effects are fully dense.
\end{corollary}

\begin{proof}
	This follows from the fact that by definition of Bernoulli randomization, $z_k \indep z_i$ for all $i \neq k$.
\end{proof}

Next we consider the scenario when we have estimates of the individual baselines so that we can subtract them from the individual outcomes when constructing our estimator. 
\begin{theorem} \label{thm:ATE_2}
	For any randomized design such that $\Prob(z_i =1 | z_k =1) = \rho_i$ for all $(k,i) \in \cE$ for some values of $\{\rho_i\}_{i \in [n]}$, the following simple estimator
	\[ \widehat{\ATE}_{-\alpha} = \frac{1}{n}\sum_{i \in [n]} \left(\frac{z_i}{\E[z_i]} - \frac{(1 - z_i)\rho_i}{\E[z_i](1-\rho_i)}\right) (Y_i(z) - \alpha_i), \]
	produces an unbiased estimate for the average treatment effect under heterogeneous additive network effects.
\end{theorem}

The condition that $\Prob(z_i =1 | z_k =1) = \rho_i$ for all $(k,i) \in \cE$ imposes symmetry in the randomized assignment amongst the neighbors of an individual. Completely randomized design or bernoulli randomization would satisfy this constraint. Alternatively a cluster based randomization would satisfy this constraint if the neighbors of a unit had equal chance of being assigned to the same versus different cluster.

\begin{proof}
	After subtracting out the baseline parameters, the constraints for unbiasedness reduce to 
	\begin{itemize}
		\item $\beta$: for all $i \in [n]$, $w_i \E[z_i] = \frac{1}{n}$,
		\item $\gamma$: for all $(k,i) \in \cE$, $w_i \E[z_i z_k] + v_i \E[(1-z_i) z_k] = 0$.
	\end{itemize}
	Satisfying the constraints arising from $\beta$ and $\gamma$ results in
	\[w_i = \frac{1}{n \E[z_i]} \text{ and } v_i = - \frac{\E[z_i z_k]}{n \E[z_i]\E[(1-z_i)z_k]} \text{ for all } (k,i) \in \cE.\]
	In order to ensure that such a valid $v_i$ exists, we would need that there exists values $\rho_i$ such that $\Prob(z_i =1| z_k=1) = \rho_i$ for all $(k,i) \in \cE$. Under this condition,
	\[\frac{\E[z_i z_k]}{\E[z_i]\E[(1-z_i)z_k]} = \frac{\rho_i}{\E[z_i](1 - \rho_i)}.\]
\end{proof}

Under completely randomized design with a budget of treating $p$ fraction of the population, $\E[z_i] = p$ and $\rho_i = \frac{pn - 1}{n - 1}$ for all $i \in [n]$. As a result,
\[ \widehat{\ATE}_{-\alpha} = \sum_{i \in [n]} \left(\frac{z_i (n-1) - pn + 1}{(1 - p)pn^2}\right) Y_i(z) - \frac{(n-1)}{n(1-p)} \left(\frac{1}{pn}\sum_{i \in [n]} z_i \alpha_i \right) + \frac{pn - 1}{(1-p)pn} \left(\frac{1}{n}\sum_{i \in [n]} \alpha_i\right). \]
Furthermore, for large enough $n$ the empirical average baseline of the treated individuals would be similar to the population average baseline, such that this estimator could be computed using the population baseline estimates by
\[ \widehat{\ATE}_{-\alpha} \approx \sum_{i \in [n]} \left(\frac{z_i (n-1) - pn + 1}{(1 - p)pn^2}\right) Y_i(z) - \frac{1}{pn^2} \sum_{i \in [n]} \alpha_i.\]

\section{Estimation the Average Interference Effect}

The average interference effect (AIE), also referred to as the network interference effect, measures the average difference in outcomes of individuals that is caused only due to network effects but not their own direct treatment effects. It is formally defined as
\[\AIE := \frac{1}{n}\sum_{i \in [n]} (Y_i({\bf e}_{[n] \setminus \{i\}}) - Y_i({\bf 0}))
= \frac{1}{n}\sum_{(k,i) \in \cE} \gamma_{ki},\]
where ${\bf e}_{\mathcal{S}}$ denotes the $n$-dimensional vector where ${\bf e}_{\cS}(x) = 1$ for $x \in \cS$ and ${\bf e}_{\cS}(x) = 0$ for $x \notin \cS$.

\begin{theorem} \label{thm:AIE}
	Under heterogeneous additive network effects, there does not exist an unbiased individually weighted linear estimator for the average interference effect.
\end{theorem}

\begin{proof}
	Under the heterogeneous additive network effects model, an individually weighted linear estimator is unbiased for the average interference effect only if $\E[\est(\bw,\bv)] = \frac{1}{n} \sum_{(k,i) \in \cE} \gamma_{ki}$ is satisfied for any configuration of $\{\alpha_i\}_{i \in [n]}, \{\beta_i\}_{i \in [n]}$, and $\{\gamma_{ki}\}_{(k,i) \in \cE}$. This requirement results in the following $2n + |\cE|$ constraints, which arise from matching coefficients for each of the parameters,
	\begin{itemize}
		\item $\alpha$: for all $i \in [n]$, $w_i \E[z_i] + v_i \E[1 - z_i]= 0$,
		\item $\beta$: for all $i \in [n]$, $w_i \E[z_i] = 0$,
		\item $\gamma$: for all $(k,i) \in \cE$, $w_i \E[z_i z_k] + v_i \E[(1-z_i) z_k] = \frac{1}{n}$.
	\end{itemize}
	The first two constraints together require that the weights are all zero, i.e. $w_i = 0$ and $v_i = 0$. However, with zero weights it is impossible to satisfy the constraint arising from $\gamma$.
\end{proof}

Next we consider the scenario when we have estimates of the individual baselines so that we can subtract them from the measured outcomes when constructing our estimator. 
\begin{theorem} \label{thm:AIE_2}
	For any randomized design such that $\Prob(z_k = 1 | z_i = 0) = \rho_i$ for all $(k,i) \in \cE$ for some values of $\{\rho_i\}_{i \in [n]}$, the following simple estimator 
	\[ \widehat{\AIE}_{-\alpha} = \frac{1}{n}\sum_{i \in [n]} \frac{(1 - z_i)}{\rho_i \E[1-z_i]} (Y_i(z) - \alpha_i), \]
	produces an unbiased estimate for the average interference effect under heterogeneous additive network effects.
\end{theorem}

\begin{proof}
	After subtracting out the baseline parameters, the constraints for unbiasedness reduce to 
	\begin{itemize}
		\item $\beta$: for all $i \in [n]$, $w_i \E[z_i] = 0$,
		\item $\gamma$: for all $(k,i) \in \cE$, $w_i \E[z_i z_k] + v_i \E[(1-z_i) z_k] = \frac{1}{n}$.
	\end{itemize}
	Satisfying the constraints arising from $\beta$ and $\gamma$ results in
	\[w_i = 0 \text{ and  } v_i = \frac{1}{n \E[(1-z_i)z_k]} \text{ for all } (k,i) \in \cE.\]
	In order to ensure that such a valid $v_i$ exists, we would need that there exists some value $\rho_i$ for which $\Prob(z_k =1 | z_i=0) = \rho_i$ for all $(k,i) \in \cE$. Under this condition,
	\[\E[(1-z_i)z_k] = \rho_i \E[1-z_i].\]
\end{proof}

The condition that $\Prob(z_k = 1 | z_i = 0) = \rho_i$ for all $(k,i) \in \cE$ imposes symmetry in the randomized assignment amongst the neighbors of an individual. Completely randomized design or bernoulli randomization would satisfy this constraint. Alternatively a cluster based randomization would satisfy this constraint if the cluster assignments were also randomized, and neighbors of a unit had equal chance of being assigned to the same versus different cluster.

Under completely randomized design with a budget of treating $p$ fraction of the population, $\E[z_i] = p$ and $\rho_i = \frac{pn}{n - 1}$ for all $i \in [n]$. Furthermore, for large enough $n$ the empirical average baseline of the treated individuals would be similar to the population average baseline, such that this estimator could be computed using the population baseline estimates by
\begin{align*}
    \widehat{\AIE}_{-\alpha} &= \frac{n-1}{(1-p) pn^2}\sum_{i \in [n]} (1 - z_i) Y_i(z) - \frac{n-1}{(1-p) pn^2}\sum_{i \in [n]} (1 - z_i)\alpha_i \\
    &\approx \frac{n-1}{(1-p) pn^2}\sum_{i \in [n]} (1 - z_i) Y_i(z) - \frac{n-1}{ pn^2}\sum_{i \in [n]} \alpha_i.
\end{align*} 

\section{Variance for General Linear Estimators}

We can compute the variance of an individually weighted linear weighted estimator for commonly used randomizations as a function of the weights $w_i$ and $v_i$. This could help in choosing the randomization for a given estimator, or choosing the weights to balance between bias and variance for a fixed randomization and estimand. As the expressions given are for a general estimator, this can be used to compute variance for the above presented estimators for the ATE and AIE.

Consider the individually weighted linear weighted estimator of the form 
\[\est(\bw,\bv) = \sum_{i \in [n]} (w_i z_i Y_i(\bz) + v_i (1-z_i) Y_i(\bz)),\]
where $w_i$ and $v_i$ are not functions of the treatment vector $\bz$. When we have baseline estimates available, we would instead subtract them from the outcomes when designing the estimator resulting in 
\[ \est_{-\alpha}(\bw,\bv) = \textstyle\sum_{i \in [n]} \left(w_i z_i + v_i (1 - z_i)\right) (Y_i(\bz) - \alpha_i).\]
The variance of $\est_{-\alpha}(\bw,\bv)$ will in fact follow from our calculations of the variance of $\est(\bw,\bv)$ with the simplifying condition that all baseline parameters $\alpha_i$ will be set to zero in the variance calculations since we have already subtracted them from the outcomes in the estimator $\est_{-\alpha}(\bw,\bv)$. As such we provide the calculations for the variance of $\est(\bw,\bv)$ as this is strictly more general.
We define the following expressions
\[L_i = (w_i - v_i) \alpha_i + w_i \beta_i +  \sum_{k \in [n]} v_k \gamma_{ik} \Ind((i,k) \in \cE) \]
\[H_{ij} = (w_i - v_i) \gamma_{ji} \Ind((j,i) \in \cE) + (w_j - v_j) \gamma_{ij} \Ind((i,j) \in \cE).\]

By expanding the expressions for $Y_i(\bz)$ from the heterogeneous additive outcomes model and rearranging terms, we can rewrite the estimator in terms of $L_i$ and $H_{ij}$ according to
\begin{align*}
	\est(\bw,\bv) &= \sum_{i \in [n]} w_i z_i (\alpha_i + \beta_i + \sum_{k \in [n]} \gamma_{ki} \Ind((k,i) \in \cE) z_k) 
	+ \sum_{i \in [n]} v_i (1-z_i) (\alpha_i + \sum_{k \in [n]} \gamma_{ki} \Ind((k,i) \in \cE) z_k)   \\
	&= \sum_{i \in [n]} v_i \alpha_i
	+ \sum_{i \in [n]} ((w_i - v_i) \alpha_i + w_i \beta_i) z_i  
	+ \sum_{i \in [n]} \sum_{k \in [n]} (w_i - v_i) \gamma_{ki} \Ind((k,i) \in \cE) z_i z_k  \\
	&\qquad+ \sum_{i \in [n]} v_i (\sum_{k \in [n]} \gamma_{ki} \Ind((k,i) \in \cE) z_k)   \\
	&= \sum_{i \in [n]} v_i \alpha_i 
	+ \sum_{i \in [n]} ((w_i - v_i) \alpha_i + w_i \beta_i +  \sum_{k \in [n]} v_k \gamma_{ik} \Ind((i,k) \in \cE)) z_i \\
	&\qquad + \sum_{i < j \in [n]^2} ((w_i - v_i) \gamma_{ji} \Ind((j,i) \in \cE) + (w_j - v_j) \gamma_{ij} \Ind((i,j) \in \cE)) z_i z_j   \\
	&= \sum_{i \in [n]} v_i \alpha_i 
	+ \sum_{i \in [n]} L_i z_i + \sum_{i < j \in [n]^2} H_{ij} z_i z_j.
\end{align*}

As a result the variance of $\est(\bw,\bv)$ is given by
\begin{align*}
	\Var[\est(\bw,\bv)]
	&= \Var[\sum_{i \in [n]} L_i z_i]  + 2 \Cov[\sum_{i \in [n]} L_i z_i,\sum_{i < j \in [n]^2} H_{ij} z_i z_j] + \Var[\sum_{i < j \in [n]^2} H_{ij} z_i z_j]   \\
	&= \sum_{i,j \in [n]^2} L_i L_j \Cov[z_i, z_j] + 2 \sum_{i \in [n]} \sum_{j < k \in [n]^2} L_i H_{jk} \Cov[z_i, z_j z_k]  \\
	&\qquad+ \sum_{i < j \in [n]^2} \sum_{k < \ell \in [n]^2} H_{ij} H_{k \ell} \Cov[z_i z_j, z_k z_\ell]   \\
	&= \sum_{i,j \in [n]^2} L_i L_j \Cov[z_i, z_j] + 2 \sum_{i \in [n]} \sum_{j < k \in [n]^2} L_i H_{jk} \Cov[z_i, z_j z_k]   \\
	&\qquad+ \sum_{i < j \in [n]^2} \sum_{k < \ell \in [n]^2} H_{ij} H_{k \ell} \Cov[z_i z_j, z_k z_\ell]   \\
	&= \sum_{i,j \in [n]^2} L_i L_j \Cov[z_i, z_j] + 2 \sum_{i \in [n]} \sum_{j < k \in [n]^2} L_i H_{jk} \Cov[z_i, z_j z_k]   \\
	&\qquad+ \sum_{i < j \in [n]^2} \sum_{k < \ell \in [n]^2} H_{ij} H_{k \ell} \Cov[z_i z_j, z_k z_\ell] .
\end{align*}

For specific randomized designs, we simply plug in the expressions for the moments of the assignment vector $\bz$.

\subsection{Completely Randomized Design}

Consider the completely randomized design, which generates the treatment assignment vector $\bz$ by selecting a subset of $pn$ units to treat uniformly at random out of the size $n$ population. The second, third, and fourth moments take value:

\begin{align*}
	\Cov[z_i, z_j] &= \begin{cases}
		p (1-p) &\text{ if } i = j \\
		- \frac{p(1-p)}{n-1} &\text{ if } i \neq j
	\end{cases} \\
	\Cov[z_i, z_j z_k] &= \begin{cases}
		- 2 p(1-p) (\frac{np-1}{(n-1)(n-2)}) &\text{ if } i \notin \{j,k\}, j<k \\
		p(1-p)(\frac{np-1}{n-1}) &\text{ if } i \in \{j,k\}, j < k
	\end{cases} \\
	\Cov[z_i z_j, z_k z_{\ell}] &= \begin{cases}
		p (\frac{np-1}{n-1}) (1 - p (\frac{np-1}{n-1})) &\text{ if } (i,j) = (k,\ell), i < j \\
		p(\frac{np-1}{n-1}) (\frac{np-2}{n-2} - p(\frac{np-1}{n-1})) &\text{ if } |\{i,j,k,\ell\}| = 3 \\
		p(\frac{np-1}{n-1}) (\frac{(np-2)(np-3)}{(n-2)(n-3)} - p(\frac{np-1}{n-1})) &\text{ if } |\{i,j,k,\ell\}| = 4
	\end{cases}
\end{align*}

The variance calculations follow by plugging the expressions for the moments and rearranging the expression algebraically,
\begin{align*}
	\Var[\est(\bw,\bv)]
	&= (1-p) p \frac{n^2}{n-1} \big(\frac{1}{n} \sum_{i} L_i^2 - \big(\frac{1}{n}\sum_{i} L_i\big)^2\big) \\
	&\qquad+ \frac{(1-p) p (n p -1)}{(n-1)(n-2)}\big( n \sum_{j < k \in [n]^2} (L_j + L_k) H_{jk}  - 2 \sum_{i} L_i \sum_{j < k \in [n]^2} H_{jk} \big) \\
	&\qquad+ p \big(\frac{np-1}{n-1}\big)\big(\big(\frac{np-2}{n-2}\big)\big(\frac{np-3}{n-3}\big) - p \big(\frac{np-1}{n-1}\big)\big) \big(\sum_{i < j \in [n]^2} H_{ij}\big)^2 \\
	&\qquad+ \frac{np(1-p)(np-1)(np-2)}{(n-1)(n-2)(n-3)} \sum_{i \in [n]} \big(\sum_{j \neq i \in [n]} H_{ij}\big)^2  \\
	&\qquad+ \frac{np(1-p)(np-1)}{(n-1)(n-2)} \big(1-\frac{np-2}{n-3}\big)\sum_{i<j \in [n]^2} H_{ij}^2 .
\end{align*}

\subsection{Cluster RD}

Consider the cluster randomized design, which partitions the population into clusters. Each cluster is either fully treated or fully control. The treatment assignment vector is generated by selecting a subset of $pT$ clusters to treat uniformly at random amongst the $T$ clusters.

We can analyze cluster RD using the equations from CRD. Let $z' \in \{0,1\}^T$ denote the cluster treatment vector such that $z_i = z'_{\pi(i)}$. Then the estimator can be written as a sum of $z'_{\tau}$ for $\tau \in T$, in particular
\[\sum_{i \in [n]} L_i z_i
= \sum_{i \in [n]} L_i z_i \sum_{\tau \in [T]} \Ind(\pi(i) = \tau) 
= \sum_{\tau \in [T]} \left(\sum_{i: \pi(i) = \tau} L_i\right) z'_{\tau}\]
and
\begin{align*}
	\sum_{i < j \in [n]^2} H_{ij} z_i z_j
	&= \sum_{i < j \in [n]^2} H_{ij} z_i z_j \sum_{\tau \leq \tau' \in [T]^2} \Ind(\pi(i) = \tau, \pi(j) = \tau' \text{ or } \pi(i) = \tau', \pi(j) = \tau) \\
	&= \sum_{\tau \in [T]} \left(\sum_{i < j: \pi(i) = \pi(j) = \tau} H_{ij} \right) z'_{\tau}
	+ \sum_{\tau < \tau' \in [T]^2} \left(\sum_{i < j: \{\pi(i),\pi(j)\} = \{\tau,\tau'\}} H_{ij} \right) z'_{\tau} z'_{\tau'}.
\end{align*}

If we define the notation
\[L'_{\tau} = \sum_{i: \pi(i) = \tau} L_i 
~~\text{ and }~~ H'_{\tau \tau'} = \sum_{i < j: \{\pi(i),\pi(j)\} = \{\tau,\tau'\}} H_{ij},\]
it follows that
\[\est(\bw,\bv)
= \sum_{i \in [n]} v_i \alpha_i 
+ \sum_{\tau \in [T]} (L'_{\tau} + H'_{\tau \tau}) z'_{\tau} + \sum_{\tau < \tau' \in [T]^2} H'_{\tau \tau'} z'_{\tau} z'_{\tau'}.\]

We can see that the estimator looks essentially the same as an estimator which operates at the level of the clusters with the modified parameters $L'$ and $H'$. Therefore, the variance can be computed by substituting the relevant quantities into the variance calculations for completely randomized design. 
\begin{align*}
	\Var[\est(\bw,\bv)]
	&= (1-p) p \frac{T^2}{T-1} \left(\frac{1}{T} \sum_{\tau} (L'_{\tau}+H'_{\tau \tau})^2 - \left(\frac{1}{T}\sum_{\tau} L'_{\tau}\right)^2\right) \\
	&\qquad+ \frac{T (1-p) p (T p -1)}{(T-1)(T-2)} \sum\limits_{\tau < \tau' \in [T]^2} (L'_{\tau} + H'_{\tau \tau} + L'_{\tau'} + H'_{\tau' \tau'}) (H'_{\tau, \tau'} + H'_{\tau', \tau}) \\
	&\qquad- \frac{2 (1-p) p (T p - 1)}{(T-1)(T-2)} \sum_{\tau''} (L'_{\tau''} + H'_{\tau'' \tau''}) \sum_{\tau < \tau' \in [T]^2} (H'_{\tau, \tau'} + H'_{\tau', \tau}) \\
	&\qquad+ p \big(\frac{Tp-1}{T-1}\big)\big(\big(\frac{Tp-2}{T-2}\big)\big(\frac{Tp-3}{T-3}\big) - p \big(\frac{Tp-1}{T-1}\big)\big) (\sum\limits_{\tau < \tau' \in [T]^2} (H'_{\tau \tau'} + H'_{\tau' \tau}))^2\\
	&\qquad+ \frac{Tp(1-p)(Tp-1)(Tp-2)}{(T-1)(T-2)(T-3)} \sum_{\tau \in [T]} \big(\sum_{\tau' \neq \tau \in [T]} (H'_{\tau \tau'} + H'_{\tau' \tau}) \big)^2  \\
	&\qquad+ \frac{Tp(1-p)(Tp-1)}{(T-1)(T-2)} \left(1-\frac{Tp-2}{T-3}\right)\sum_{\tau < \tau' \in [T]^2} (H'_{\tau \tau'} + H'_{\tau' \tau})^2.
\end{align*}

\subsection{Cluster stratified RD}

Consider the cluster stratified randomized design, which partitions the population into clusters, and within each cluster a specified fraction of the units are treated uniformly at random amongst all units in that cluster. The assignments in different clusters are fully independent. We can rearrange the estimator and use the independence of assignments across different clusters to simplify the variance calculations. Using the definitions of $L_i$ and $H_{ij}$,
\[\est(\bw,\bv)
= \sum_{i \in [n]} v_i \alpha_i
+ \sum_{\tau \in [T]} \left(\sum_{i: \pi(i) = \tau} L_i z_i + \sum_{i < j: \pi(i) = \pi(j) = \tau} H_{ij} z_i z_j\right) 
+ \sum_{\tau < \tau' \in [T]^2} \left(\sum_{i < j: \{\pi(i),\pi(j)\} = \{\tau,\tau'\}} H_{ij} z_i z_j\right).\]
Let us define the notation
\[A_{\tau} = \sum_{i: \pi(i) = \tau} L_i z_i + \sum_{i < j: \pi(i) = \pi(j) = \tau} H_{ij} z_i z_j 
~~\text{ and }~~ B_{\tau \tau'} = \sum_{i < j: \{\pi(i),\pi(j)\} = \{\tau,\tau'\}} H_{ij} z_i z_j.\]
Because the assignments are independent across different clusters, the variance calculation reduces to
\begin{align*}
	\Var[\est(\bw,\bv)]
	&= \sum_{\tau \in [T]} \Var[A_{\tau}]
	+ \sum_{\tau \in [T]} \sum_{\tau' \neq \tau \in [T]} \Cov[A_{\tau}, B_{\tau \tau'}] 
	+ \sum_{\tau < \tau' \in [T]^2} \Var[B_{\tau \tau'}] \\
	&\qquad+  \sum_{\tau} \sum_{\tau' \neq \tau'', \tau \notin \{\tau' \tau''\}} \Cov[B_{\tau \tau'}, B_{\tau \tau''}]. 
\end{align*}

Observe that $\Var\left[A_{\tau}\right] $ is the same as the variance calculations for completely randomized design restricted to the cluster $\tau$, such that
\begin{align*}
	\Var\left[A_{\tau}\right]
	&= (1-p_{\tau}) p_{\tau} \frac{n_{\tau}^2}{n_{\tau}-1} \big(\frac{1}{n_{\tau}} \sum_{i: \pi(i) = \tau} L_i^2 - \big(\frac{1}{n_{\tau}}\sum_{i: \pi(i) = \tau} L_i\big)^2\big) \\
	&\qquad+ \frac{2(1-p_{\tau}) p_{\tau} (n_{\tau} p_{\tau} -1)}{(n_{\tau}-1)(n_{\tau}-2)}\big( n_{\tau} \sum_{i<j: \pi(i) = \pi(j) = \tau} L_i H_{ij}  - \sum_{k: \pi(k) = \tau} L_k \sum_{i<j: \pi(i) = \pi(j) = \tau} H_{ij} \big) \\
	&\qquad+ p_{\tau} \left(\frac{n_{\tau}p_{\tau}-1}{n_{\tau}-1}\right)\big(\big(\frac{n_{\tau}p_{\tau}-2}{n_{\tau}-2}\big)\big(\frac{n_{\tau}p_{\tau}-3}{n_{\tau}-3}\big) - p_{\tau} \big(\frac{n_{\tau}p_{\tau}-1}{n_{\tau}-1}\big)\big) \big(\sum_{i<j: \pi(i) = \pi(j) = \tau} H_{ij}\big)^2\\
	&\qquad+ \frac{n_{\tau}p_{\tau}(1-p_{\tau})(n_{\tau}p_{\tau}-1)(n_{\tau}p_{\tau}-2)}{(n_{\tau}-1)(n_{\tau}-2)(n_{\tau}-3)} \sum_{i: \pi(i) = \tau} \big(\sum_{j \neq i: \pi(j) = \tau} H_{ij}\big)^2  \\
	&\qquad+ \frac{n_{\tau}p_{\tau}(1-p_{\tau})(n_{\tau}p_{\tau}-1)}{(n_{\tau}-1)(n_{\tau}-2)} \left(1-\frac{n_{\tau}p_{\tau}-2}{n_{\tau}-3}\right)\sum_{i<j: \pi(i) = \pi(j) = \tau} H_{ij}^2 .
\end{align*}

For $\tau \neq \tau'$,
\begin{align*}
	\Cov[A_{\tau}, B_{\tau \tau'}]
	&= \frac{p_{\tau} (1-p_{\tau}) p_{\tau'}}{n_{\tau}-1} \Big(n_{\tau} \sum_{i,k: \pi(i) = \tau,\pi(k) - \tau'} L_i H_{ik}
	- \sum_{i: \pi(i) = \tau} L_i \sum_{h,k: \pi(h) = \tau, \pi(k) = \tau'} H_{hk} \\
	&\qquad \qquad \qquad \qquad \qquad + \frac{n_{\tau} (n_{\tau} p_{\tau}-1)}{n_{\tau}-2} \sum_{i<j: \pi(i) = \pi(j) = \tau} \sum_{k: \pi(k) = \tau'} H_{ij} (H_{ik} + H_{jk}) \\
	&\qquad \qquad \qquad  \qquad \qquad- \frac{2(n_{\tau} p_{\tau}-1)}{n_{\tau} - 2} \sum_{i<j: \pi(i) = \pi(j) = \tau} H_{ij} \sum_{h,k: \pi(h) = \tau, \pi(k) = \tau'} H_{hk} \\
	&\qquad \qquad \qquad \qquad \qquad- \frac{2(n_{\tau} p_{\tau}-1)}{n_{\tau} - 2} \sum_{i<j: \pi(i) = \pi(j) = \tau} H_{ij} \sum_{h,k: \pi(h) = \tau, \pi(k) = \tau'} H_{hk} \Big).
\end{align*}

For $\tau \neq \tau'$,
\begin{align*}
	&\Var[B_{\tau \tau'}] \\
	&= \frac{n_{\tau}p_{\tau}(1-p_{\tau})n_{\tau'}p_{\tau'}(1-p_{\tau'})}{(n_{\tau}-1)(n_{\tau'}-1)} \sum_{h,k: \pi(h) = \tau, \pi(k) = \tau'} H_{hk}^2 + \frac{p_{\tau} p_{\tau'} n_{\tau'}(1-p_{\tau'})(n_{\tau} p_{\tau} -1)}{(n_{\tau}-1)(n_{\tau'}-1)} \sum_{k: \pi(k) = \tau'} \left(\sum_{i: \pi(i) = \tau} H_{ik} \right)^2 \\
	&\qquad+ \frac{p_{\tau} p_{\tau'} n_{\tau}(1-p_{\tau})(n_{\tau'} p_{\tau'} - 1)}{(n_{\tau}-1)(n_{\tau'}-1)} \sum_{i: \pi(i) =\tau} \left(\sum_{k: \pi(k) = \tau'} H_{ik} \right)^2 \\
	&\qquad+ p_{\tau} p_{\tau'} \left(\left(\frac{n_{\tau'} p_{\tau'} - 1}{n_{\tau'} -1}\right) \left(\frac{n_{\tau}p_{\tau} - 1}{n_{\tau}-1}\right) - p_{\tau} p_{\tau'}\right) \left(\sum_{i,j: \pi(i) = \tau, \pi(j) = \tau'} H_{ij} \right)^2.
\end{align*}

For a distinct triple $\tau \neq \tau' \neq \tau''$,
\begin{align*}
	&\Cov[B_{\tau \tau'}, B_{\tau \tau''}] \\
	&= \frac{p_{\tau} (1-p_{\tau}) p_{\tau'} p_{\tau''}}{n_{\tau}-1}\Big(&n_{\tau} \sum_{i,k,h: \pi(i) = \tau, \pi(k) = \tau', \pi(h) = \tau''} H_{ik} H_{ih} 
	- \sum_{i, k: \pi(i) = \tau, \pi(k) = \tau'} H_{ik} \sum_{j,h: \pi(j) = \tau, \pi(h) = \tau''} H_{ih}\Big).
\end{align*}
The final variance of the estimator results from combining these expressions together.

\end{document}